\documentclass[12pt]{iopart}

\expandafter\let\csname equation*\endcsname\relax
\expandafter\let\csname endequation*\endcsname\relax
  
\usepackage{amsmath}
\usepackage{amsfonts}
\usepackage{amssymb}
\usepackage{amsthm}
\usepackage{stmaryrd}
\usepackage{bm}
\usepackage{cite}
\usepackage[english]{babel}
\usepackage{dblaccnt}
\usepackage{accents}
\usepackage{graphicx}
\usepackage{subfig}
\usepackage{amscd}
\usepackage[mathscr]{euscript}
\usepackage{graphicx, color,float}
\usepackage[toc,page]{appendix}
\usepackage{tikz}
\usepackage{comment}

\newtheorem{defi}{Definition}[section]
\newtheorem{lemma}[defi]{Lemma}
\newtheorem{theorem}[defi]{Theorem}

\newtheorem{asp}{Assumption}  

\newtheorem{definition}[defi]{Definition}

\newcommand{\mfied}[1]{{\color{black} #1}}
\newcommand{\opeh}{{\cal H}}
\newcommand{\opei}{{\mathrm{I}}}
\newcommand{\opeg}{{\mathrm{G}}}
\newcommand{\open}{{\mathrm{N}}}
\newcommand{\opef}{{\mathrm{F}}}

\newcommand{\opet}{{\mathrm{T}}}

\newcommand{\Odual}[1]{(#1)^{*}}
\newcommand{\alpm}[1]{{\alpha_\sper(#1)}}
\newcommand{\betm}[1]{{\beta_\sper(#1)}}
\newcommand{\olbetm}[1]{{\ol \beta_\sper(#1)}}
\newcommand{\uip}{u^{i,+}}
\newcommand{\uim}{u^{i,-}}
\newcommand{\uipm}{u^{i,\pm}}
\newcommand{\fg}{{f}}
\newcommand{\ftd}{{\varphi}}
\newcommand{\ftn}{{\psi}}

\newcommand{\greq}[1]{{\Phi_q(\cdot- #1)}}

\newcommand{\sepmgre}[1]{{\widehat{\Phi}^\pm(\cdot- #1)}}
\newcommand{\sepmgreq}[1]{{\widehat{\Phi}_q^\pm(\cdot - #1)}}
\newcommand{\sepgre}[1]{{\widehat{\Phi}^+(\cdot - #1)}}
\newcommand{\sepgreq}[1]{{\widehat{\Phi}^+_q(\cdot;#1)}}
\newcommand{\coefpmgre}[2]{{\widehat{\Phi}^\pm(#2 - #1)}}
\newcommand{\coefpmgreq}[2]{{\widehat{\Phi}_q^\pm(#2 - #1)}}

\newcommand{\omehm}{{\Omega^h_M}}

\newcommand{\gamhmp}{{\Gamma^h_M}}
\newcommand{\gamhmm}{{\Gamma^{-h}_M}}

\newcommand{\boul}{{M_L^-}}
\newcommand{\bour}{{M_L^+}}
\newcommand{\spahsm}[1]{{H^{#1}_{\sper}(\omehm)}}

\newcommand{\spatracep}[1]{{H^{#1}_\sper(\gamhmp)}}
\newcommand{\spatracem}[1]{{H^{#1}_\sper(\gamhmm)}}

\newcommand{\raycoefp}[2]{{\widehat{#1}^+(#2)}}
\newcommand{\raycoefm}[2]{{\widehat{#1}^-(#2)}} 
\newcommand{\raycoefpm}[2]{{\widehat{#1}^\pm(#2)}}

\newcommand{\hdualcoefpm}[1]{{\widehat{#1}^\pm}}

\newcommand{\seth}{{H}}
\newcommand{\inc}{{\mathrm{inc}}}

\newcommand{\aher}{a}
\newcommand{\zg}{{z}}
\newcommand{\eg}{u}
\newcommand{\wgg}{{w}}

\newcommand{\vgg}{{v}}
\newcommand{\fug}{{{\varphi}}}
\newcommand{\xg}{{x}}

\newcommand{\Rang}{\mathcal{R}}

\newcommand{\Zd}{\Z^{d-1}}
\newcommand{\domp}{\omega}

\newcommand{\fundnp}[1]{{\Phi(n_p; #1)}}
\newcommand{\tfundnp}[1]{{\widetilde{\Phi}(#1)}}
\newcommand{\wpr}{{w}}
\newcommand{\twpr}{{\widetilde{\wpr}}}
\newcommand{\opes}[1]{{\widetilde{\mathcal{S}}_{#1}}}

\newcommand{\cdomp}{{\domp^c}}

\newcommand{\fc}{{\widetilde{f}}}
\renewcommand{\BB}{{\mathcal{B}}}

\newcommand{\Dint}{{\Lambda}}
\newcommand{\Dpint}{{\mathcal{O}}}
\newcommand{\Dpout}{{\mathcal{O}^c}}
\newcommand{\Dpoutp}{{\mathcal{O}^c_p}}

\newcommand{\Dtot}{{\widehat{D}}}

\renewcommand{\i}{\mathrm{i}}
\renewcommand{\d}[1]{\,\mathrm{d}#1 \,}

\renewcommand{\Re}{\mathrm{Re}\,}
\renewcommand{\Im}{\mathrm{Im}\,}

\newcommand{\ol}[1]{\overline{#1}}

\renewcommand{\epsilon}{\varepsilon}
\newcommand{\loc}{\mathrm{loc}}

\newcommand*{\N}{\ensuremath{\mathbb{N}}}
\newcommand*{\Z}{\ensuremath{\mathbb{Z}}}
\newcommand*{\R}{\ensuremath{\mathbb{R}}}

\newcommand{\normM}[1]{{\llbracket {#1} \rrbracket}}

\newcommand{\interd}[2]{\ensuremath{\llbracket {#1}, {#2} \rrbracket}}

\newcommand{\I}{\mathcal{I}}

\newcommand{\CC}{\ensuremath{\mathbb{C}}}

\newcommand{\X}{\mathbf{X}}

\newcommand{\sper}{\#}

\newcommand{\dsp}{\displaystyle}

\begin{document}

\title[Analysis and Applications of Interior Transmission Problems]{New Interior Transmission Problem Applied to a Single Floquet-Bloch Mode Imaging of Local Perturbations in Periodic Media}

\author{Fioralba Cakoni$^1$,    Houssem Haddar$^2$, Thi-Phong Nguyen$^1$}

\address{$^1$ Department of Mathematics, Rutgers University, 110 Frelinghuysen Road,  Piscataway, NJ 08854-8019, USA.}
\address{$^2$ INRIA, Ecole Polytechnique (CMAP) and Universit\'{e} Paris Saclay,
Route de Saclay, 91128 Palaiseau Cedex, France.}
\ead{fc292@math.rutgers.edu,  Houssem.haddar@inria.fr,  tn242@math.rutgers.edu}
\vspace{10pt}
\begin{indented}
\item[] 
\end{indented}

\begin{abstract}

This paper considers the imaging  of  local  perturbations of an infinite penetrable periodic layer. A cell of this  periodic layer consists of  several bounded inhomogeneities situated in a known homogeneous media. We use  \mfied{a differential linear sampling method} to reconstruct the support of perturbations without using the Green's function of the periodic layer nor reconstruct the periodic background inhomogeneities. The justification of this imaging method relies on the well-posedeness of a nonstandard interior transmission problem, which until now was an open problem except for the special case when the  local perturbation didn't intersect the background inhomogeneities. The analysis of this new interior transmission problem is the main focus of this paper.  We then complete the justification of  our inversion  method and present some numerical examples that confirm the theoretical behavior of the differential indicator function  determining the reconstructable regions in the periodic layer.
\end{abstract}

\section{Introduction}

Nondestructive testing of period media is an important problem with grown interest since periodic material are part of many fascinating engineering structures with many technological use such as nanograss.  In many situation the periodicity of the  healthy periodic material is complicated or difficult to model mathematically, hence computing its Green's function is computationally expensive or even impossible. On the other hand, when looking for flows in such complex media, the option of reconstructing everything, i.e. both periodic structure and defects, may not be viable. The approach used in this paper provides a criteria to reconstruct  the support of anomalies without explicitly know or reconstruct the background. The imaging method is based on the generalized linear sampling method which was first introduced in  \cite{Audib2015a}, \cite{Audib2014}. This method falls in the class of qualitative approaches to inverse scattering. We refer the reader to \cite{Kirsc2008} and \cite{CCH} for a description of various \mfied{aspects} of such approaches. Qualitative methods have been applied to the imaging of many periodic structure, see  \cite{Arens2010},  \cite{Arens2005},  \cite{Bourg2014}, \cite{Elsch2011a}, \cite{Lechl2013b},   \cite{Thi-Phong3},  \cite{armin},  \cite{nguye2012} for a sample of work. In the case of our problem, we use \mfied{an adapted version of so-called differential linear sampling method} which process the measured data against the data coming from the healthy background. The idea of using differential measurements \mfied{for sampling methods} was first  introduced in \cite{Audib2015} where the response of the background was measured, and was adapted to the case of \mfied{ locally perturbed periodic layers} in  \cite{Thi-Phong3}, \cite{tpnguyen}. For the latter, \mfied{the response of the periodic background does not need to measured. It is replaced by the extraction of measurements associated with a single Floquet-Bloch mode to encode some differential behavior for the indicator functions. This extraction requires} information only on the period \mfied{size} of the background.  The justification of this method makes essential use of a non-standard interior transmission problem whose well-posedness was open, and this limited its use to the case when the defect does not intersect the inhomogeneous components of the background. In this paper we provide sufficient conditions for solvability of this non-standard interior transmission problem which allow us to design a differential imaging function for more general location of defects. Let us introduce the problem we consider here.

\noindent
More specifically, we are concerned with nondestructive testing of a penetrable infinite layer in ${\mathbb R}^d$, $d=2,3$ which is periodic with respect to $d-1$ first variables. Let   $L_1, \cdots, L_{d - 1}$,  $L_j > 0, \
j = 1, \cdots, d-1$ denote the periods of each of these  $d-1$ variables, respectively. The $d-1$ periodic refractive index of this periodic layer,  denoted here by $n_p$, from physical consideration is a bounded function,  has  positive real part $\Re(n_p)$ and nonnegative imaginary part $\Im(n_p)\geq 0$. Furthermore for simplicity we assume that this periodic layer is embedded in a homogeneous background with refractive index normalized to one, i.e. $n_p=1$  for $|x_d| > h$  for some fixed $h>0$. This is what we refer to as the healthy material.  We assume that one or finitely many cells of the layer are locally damaged. This means that in a compactly supported region $\omega$ (which can have multiple connected components) the refractive index differs from $n_p$. Let us call $n$ the refractive index of the damaged layer (which is not any longer periodic), i.e. $n\neq n_p$ only in $\omega$. The goal is to determine the support of the damaged region $\omega$  by using the measured scattered field outside the layer due to appropriate incident fields (to become precise later).  The challenging task however is to resolve $\omega$ without an explicit knowledge of $n_p$ (which in practice can have complicated form) nor reconstructing it, but just using the fact that $n_p$ is $d-1$ periodic with known periods $L_1, \cdots, L_{d - 1}$ under some technical restriction which will be explained in the paper.

\noindent
The paper is configured as follows. In the next section we formulate the direct and inverse problem, define the measurements operator and recall some of its properties which are essential to our imaging method.  Section 3 is devoted to introducing the near field operator corresponding to a single Floquet-Bloch mode that unable us to use a differential imaging approach. Most importantly here we study the properties of this operator which bring up the new interior transmission problem.  Section 4 is devoted to the analysis of this new interior transmission problem. In the last section we build the differential imaging function and study its behavior for various positions of defective regions. Here we provide some numerical example showing the viability of our inversion method. 

\section{Formulation of the Problem}
In this section we give a rigorous formulation of the direct and inverse scattering problem we consider here. In order to motivate the new interior transmission problem which is our main concern, we recall the \mfied{differential linear sampling method} that was first introduced in \cite{tpnguyen} (see also  \cite{Thi-Phong3}). This method recovers the support of local perturbations of a periodic layer without  needing to compute the Green's function of the periodic layer.  However to do so we must make some technical mathematical restrictions aimed to preserve some kind of periodicity for the damaged layer. In particular, we truncate  the damaged infinite periodic layer by considering $M$ periods \mfied{(with $M$ large enough to contain the defect)} and extend it periodically, yielding to a $ML:=(ML_1, \cdots, ML_{d - 1})$-periodic layer. We call again $n$ the refractive index of the $ML$-periodic extension of the truncated part.  In this section we formulate rigorously this construction where we base our inversion algorithm. We remark that it is highly desirable to remove this mathematical artifact. 
\subsection{The Direct Scattering Problem}
Here we \mfied{adopt} the notations in  \cite{Thi-Phong3}. Recall that the parameter  $L := (L_1, \cdots, L_{d - 1}) \in \R^{d-1}, \ L_j > 0, \
j = 1, \cdots, d-1$ refers to the periodicity of the media with
respect to the first $d-1$ variables and $M := (M_1, \cdots, M_{d - 1}) \in
\N^{d-1}$ refers to the number of periods in the truncated domain. A function defined in $\R^d$ is called
$L$ periodic if it is periodic with
period $L$ with respect to the $d-1$ first variables.    We consider in the following  $ML-$periodic Helmholtz
equation (vector multiplications is to be understood component
wise, i.e. $ML =  (M_1L_1, \cdots, M_{d - 1}L_{d - 1})$). In this problem, the
  total field $u$ satisfies 
\begin{equation} \label{C4Eq:Helmh}
	\left\{
	\begin{array}{lc}
		\Delta u + k^2 nu = 0 \quad \text{in} ~ \R^d, ~ d = 2,3 \\[1.5ex]
		u ~ \text{is $ML-$periodic}
	\end{array}
	\right.
\end{equation} where $k>0$ is the {\it wave number}. 
\begin{figure}[htpb]
\centerline{\begin{tikzpicture}[scale=0.9]
\draw[blue,thick,dashed] (-6.5,1.9) -- (10.5,1.9);
\draw[blue,thick,dashed] (-6.5,-1.7) -- (10.5,-1.7);
\draw[blue,thick] (-6,-3.2) -- (-6,2.2); 
\draw[blue,thick] (-2,-2) -- (-2,2.2);
\draw[blue,thick] (2,-2) -- (2,2.2); 
\draw[blue,thick] (6,-2) -- (6,2.2);
\draw[blue,thick] (10,-3.2) -- (10,2.2);


\fill[color=red!60, fill=red!50, very thick](-4.,-0.7) rectangle (-3.,0.3);
\fill[color=red!60, fill=red!50, very thick](0,-0.7) rectangle (1.,0.3);
\fill[color=red!60, fill=red!50, very thick](4.,-0.7) rectangle (5.,0.3);
\fill[color=red!60, fill=red!50, very thick](8.,-0.7) rectangle (9.,0.3); 

\fill[color=red!60, fill=red!50, very thick](-5.,0.9) circle (0.75);
\fill[color=red!60, fill=red!50, very thick](-1.,0.9) circle (0.75);
\fill[color=red!60, fill=red!50, very thick](3.,0.9) circle (0.75);
\fill[color=red!60, fill=red!50, very thick](7.,0.9) circle (0.75);

\fill[color=blue!60, fill=blue!40, very thick](-0.2,1.2) circle (0.4);

\small
\draw[blue] (0,2.3) node {$ n = n_p = 1$};
\draw[blue] (-6.4,1.4) node {$ h$};
\draw[blue] (-6.6, -1.2) node {$ -h$};
\draw[black] (-0.2,1.2) node {$ \domp$};
\draw[blue](-1.1,0.7)node {$\Dpint$};

\draw[blue] (0.4,-0.3) node {$ \Dpout$}; 

\draw[blue] (1.6,-1.4) node {$ \Omega_0$};

\draw [blue,ultra thin,<->] (-2,-1.9) -- (2,-1.9);  \draw[blue,thick,dashed] (0,-2.1) node {$L$};

\draw[black] (- 0.6,-2.7) node {$ \Dint: = \Dpint \cup \domp$,}; 

\draw[black] (2.6,-2.7) node {$ \Dtot: = \Dint \cup \Dpout$}; 

\draw [red,thin,<->] (-6,-3.2) -- (10,-3.2);  \draw[red,thick,dashed] (0,-3.6) node {$ ML$};

\end{tikzpicture}}
\caption{Sketch of the geometry for the $ML-$periodic problem}
\label{Fig:ML-periodic}
\end{figure}
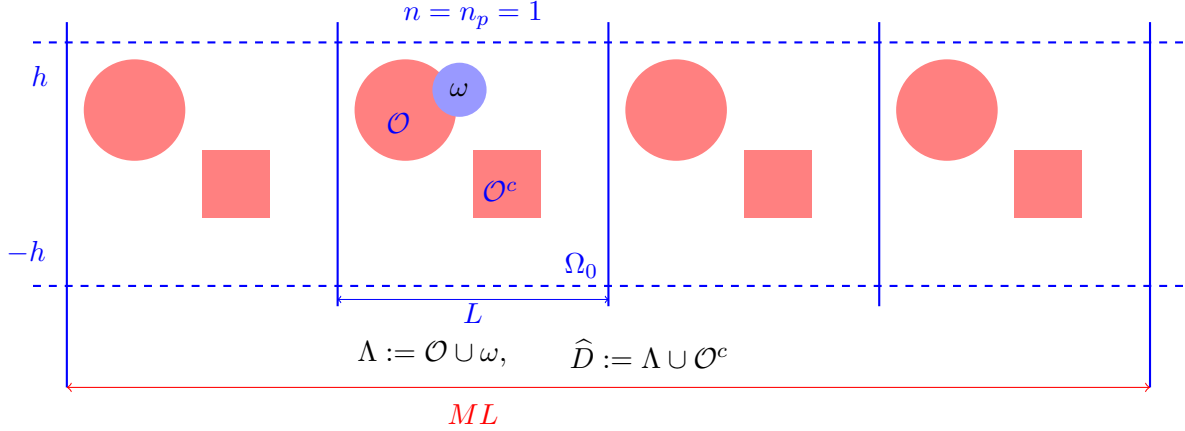
We assume that the index of refraction $n\in L^\infty(\R^d)$ satisfies $\Re(n)\geq n_0>0$,  $\Im(n)\geq 0$ and is $ML-$periodic. Furthermore  $n = n_p$ outside a compact domain $\domp$ where  $n_p \in L^\infty(\R^d)$ is $L$-periodic, and in addition there exists $h >0$ such that $n =1$ for $|x_d| > h$ (see Fig. \ref{Fig:ML-periodic}). Thanks to the $ML-$periodicity, solving equation \eqref{C4Eq:Helmh} in $\R^d$ is equivalent to solving it in the period 
 $$\Omega_M:=\bigcup_{m \in \Zd_M} \Omega_m =  \interd{\boul}{\bour} \times \R $$
  with  $\Omega_m: = \Omega_0 + mL$, $\boul:= \left(\left\lfloor-\frac{M}{2}\right\rfloor + \frac{1}{2}\right)L$, $\bour: =\left(\left\lfloor\frac{M}{2}\right\rfloor + \frac{1}{2}\right) L$, \; $ \Omega_m : =
  \interd{-\frac{L}{2} +mL}{ \frac{L}{2} +mL}\times\R$ and \; $ \Zd_M: = \{ m\in \Z^{d-1}, \textstyle \left\lfloor-\frac{M_\ell}{2}\right\rfloor+1 \leq m_\ell \leq \left\lfloor\frac{M_\ell}{2}\right\rfloor , \; \ell = 1, \ldots, d-1\}$, where we use the notation $\interd{a}{b} := [a_1, b_1] \times \cdots \times
[a_{d-1}, b_{d-1}]$ and $\lfloor\cdot \rfloor$ denotes the floor function. \mfied{We also shall use the notation $\normM{a} := |a_1\cdot a_2 \cdots a_{d-1}|$}. Without loss of generality we assume that there is a local perturbation $\domp$ located in only one period, say $\Omega_0$ (note the case when more periods are defective, the assumption holds true by grouping these cells as one cell  with different period). This problem  is treated in\cite{tpnguyen} under a strict assumption that the local perturbation does not intersect with the periodic background. In this work, we remove this assumption, and allow for the the local perturbation  to be located everywhere in $\Omega_0$.  We call $D_p$ the support of $n_p-1$ and $D = D_p \cup \domp$, note that $n=1$ outside $D$. \mfied{For the justification of our inversion method (that relies on a unique continuation argument) we make the assumption that $\R^d \setminus D$ is connected}.

We  consider down-to-up or up-to-down  incident plane waves of the form
\begin{equation}\label{inc}
	 u^{i,\pm}(x,j)=\frac{-\i}{2\,\olbetm{j}}e^{\i \alpm{j} \ol{x} \pm  \i \olbetm{j} x_d } 
\end{equation}
where 
$$\quad \textstyle{ \alpm{j} := \i\frac{2\pi}{ML}j \quad \mbox{ and } \quad  \betm{j} := \sqrt{k^2 - \alpha^2_\sper(j)}}, \quad \Im( \betm{j}) \ge 0, \quad j \in \Zd$$
and $x=(\ol{x},x_d) \in \R^{d-1}\times \R$. Then the scattered field $u^s = u- u^i$ verifies 
\begin{equation}\label{scat}
	\left\{
	\begin{array}{lc}
		\Delta u^s + k^2 nu^s = -k^2(n-1)u^i \quad \text{in} ~ \R^d, \\[1.5ex]
		u^s ~ \text{is $ML-$periodic}
	\end{array}
	\right.
\end{equation}
and we impose as a radiation condition  the Rayleigh expansions:
\begin{equation} \label{C4:RDC}
\left\{
\begin{array}{lc}
	u^s(\ol{x},x_d) = \sum_{\ell \in \Zd} \raycoefp{u^s}{\ell} e^{\i(\alpm{\ell} \ol{x} + \betm{\ell} (x_d-h))},  \quad \forall \  x_d >h, \\[1.5ex]
	u^s(\ol{x},x_d) = \sum_{\ell \in \Zd} \raycoefm{u^s}{\ell} e^{\i(\alpm{\ell}\ol{x} - \betm{\ell} (x_d+h))}, \quad \forall \ x_d <-h,
\end{array}
\right.
\end{equation}
where the Rayleigh coefficients $\raycoefpm{u^s}{\ell}$ are given by
\begin{equation}\label{RayleighCh2}
\left.
\begin{array}{l l}
\raycoefp{u^s}{\ell} := \dsp \frac{1}{|\interd{\boul}{\bour}|} \int_{\interd{\boul}{\bour}} u^{s}(\ol{x}, h) e^{ - \i \alpm{\ell} \cdot \ol{x}} \d{\ol{x}},  \\[2ex]
\raycoefm{u^s}{\ell}:= \dsp \frac{1}{|\interd{\boul}{\bour}|} \int_{\interd{\boul}{\bour}} u^{s}(\ol{x}, - h) e^{ - \i \alpm{\ell} \cdot \ol{x}} \d{\ol{x}}.
\end{array} 
\right.	
\end{equation}
We shall use the notation
 $$\omehm := \interd{\boul}{\bour} \times ]-h,h[$$
$$\gamhmp:= \interd{\boul}{\bour} \times \{h\}, \quad   \gamhmm:=
\interd{\boul}{\bour} \times \{- h\}.$$
For integer $m$, we denote by $H^m_{\sper}(\Omega_M^h)$ the restrictions to $\Omega_M^h$ of
functions  that are in $H^m_{\mathrm{loc}}(|x_d| \le h)$ and are
$ML-$periodic. The space $\spatracep{1/2}$ is then defined as the space of
traces on $\gamhmp$ of functions in $H^1_\sper(\Omega_M^h)$ and the space
$\spatracep{-1/2}$ is defined as the dual of $\spatracep{1/2}$. Similar definitions are used for $\spatracem{\pm 1/2}$.
\\
More generally for a given  $f \in L^2(\Omega^h_M)$, we  consider the following problem: Find $w \in \spahsm{1}$ satisfying 
\begin{equation}
\label{eq:w}
\Delta w + k^2 n w = k^2(1-n) f
\end{equation}
together with the Rayleigh radiation condition (\ref{C4:RDC}). Then we make the following assumption:
\begin{asp} \label{Ass:nk}
The refractive index $n$ and $k>0$ are such that  \eqref{eq:w}  with $n$ and  with $n$ replaced by  $n_p$ are both well-posed for all $f \in L^2(\Omega^h_M)$. 
\end{asp}
\noindent
We remark that the solution  $w \in \spahsm{1}$ of \eqref{eq:w} can be
extended to a function in $\Omega_M$ satisfying $\Delta \wgg + k^2 n \wgg =
k^2(1-n) f$, using the Rayleigh expansion \eqref{C4:RDC}. We denote by $\fundnp{\cdot}$ the fundamental solution to 
\begin{equation}\label{phi}
\left\{
\begin{array}{lc}
	\Delta \fundnp{\cdot}  + k^2 n_p \fundnp{\cdot} = - \delta_0, \\[1.5ex]
	\fundnp{\cdot} \  \text{is} \ ML - \text{periodic}, \\[1.5ex]
	\text{and the Rayleigh radiation condition (\ref{C4:RDC}). }
\end{array}
\right.
\end{equation}
Then $w$ has the representation as 
\begin{equation}
\label{eq:2forw}
	w(x) = -\int_{D} \Big(k^2(n_p - n)w + k^2(1 - n) f \Big)(y) \fundnp{x-y} \d{y}.
\end{equation}
For sufficient conditions that guaranty Assumption \ref{Ass:nk} we refer the reader to \cite{Thi-Phong2},  \cite{KL}, \cite{tpnguyen} and the references therein.

\subsection{The Inverse Problem}
The inversion method is based on the so-called the Generalized Linear Sampling Method, which was first introduced in  \cite{Audib2015a}, \cite{Audib2014} (see also  \cite[Chapter 2]{CCH}), augmented with the idea of \mfied{differential imaging} introduced in  \cite{Audib2015} which was  adapted to this problem in \cite{tpnguyen}.

\noindent 
As described above we have two choices of interrogating waves. If  we use  down-to-up (scaled) incident plane waves $\uip(x;j)$ defined by (\ref{inc}), then our measurements (data for the inverse problem) are  given by the Rayleigh sequences 
$$
\raycoefp{u^s}{\ell;j},  \quad  (j, \ell) \in \Zd \times \Zd,  
$$
whereas if we use up-to-down (scaled) incident plane waves $\uim(x;j)$ defined by (\ref{inc}) then our  measurements are given the Rayleigh sequences
$$
\raycoefm{u^s}{\ell;j}, \quad  (j, \ell) \in \Zd \times \Zd.
$$
These measurements  define  the so-called near field (or data) operator which is used to derive the indicator function of the defect. More specifically, let us consider the (Herglotz) operators $\opeh^{+}: \ell^2(\Zd) \rightarrow L^2(D)$ and $\opeh^{-}: \ell^2(\Zd) \rightarrow L^2(D)$ defined by 
\begin{equation} \label{2defH}
\opeh^{\pm}\aher:= \sum_{j \in \Zd} \aher(j) \uipm(\cdot; j)\big|_D, \quad \forall \, \aher = \{\aher(j)\}_{j \in \Zd} \in \ell^2(\Zd).
\end{equation}Then $\opeh^{\pm}$ is compact and its adjoint $\Odual{\opeh^{\pm}}: L^2(D) \to \ell^2(\Zd) $ is given by  \cite{Thi-Phong3}
  \begin{equation}\label{2adjointH}
\Odual{\opeh^{\pm}} \fug := \{\hdualcoefpm{\fug}(j)\}_{j \in \Zd}, \quad \mbox{where} \quad \hdualcoefpm{\fug}_j : = \int_{D} \fug(\xg) \ol{\uipm(\cdot; j)}(\xg)\d{\xg}.
\end{equation} 
Let us denote by  $\seth_{\inc}^{\pm}(D)$ the closure of the range of $\opeh^{\pm}$ in $L^2(D)$. We then consider the (compact) operator $\opeg^{\pm}:\seth_{\inc}^{\pm}(D) \rightarrow \ell^2(\Zd)$ defined by
\begin{equation} \label{defG}
{\opeg^{\pm}} (f) := \{\raycoefpm{w}{\ell}\}_{\ell \in\Zd},
\end{equation}
where $\{\raycoefpm{w}{\ell}\}_{\ell \in\Zd}$ is the Rayleigh sequence of $ \wgg \in\spahsm{1}$ the solution of \eqref{eq:w}. We now define the sampling operators $\open^{\pm}: \ell^2(\Zd) \rightarrow \ell^2(\Zd) $ by
\begin{equation}
\open^{\pm} (\aher) = {\opeg^{\pm}}  \, \opeh^\pm(\aher).
\end{equation}
By linearity of the operators ${\opeg^{\pm}}$ and $\opeh^\pm$ we also get an equivalent definition of $\open^{\pm}$ directly in terms of measurements as
\begin{equation}
[\open^{\pm} (\aher)]_\ell = \sum_{j \in \Zd} a(j) \, \raycoefpm{u^s}{\ell;j} \quad \ell \in \Zd.
\end{equation}
 The following properties of  $\opeg^{\pm}$ and and $\opeh^\pm$ are crucial to our inversion method.  To state them, we must recall the standard  {\it interior transmission problem}:
$({\eg}, \vgg) \in
L^2(D) \times L^2(D)$ such that $
{\eg}-\vgg \in H^2(D) $
 and 
\begin{equation} \label{oitp}
\left\{ \begin{array}{lll}
\Delta {\eg} + k^2 n {\eg} = 0 \quad& \mbox{ in } \; D, 
\\[6pt]
\Delta \vgg + k^2 \vgg = 0  \quad &\mbox{ in } \; D,
\\[6pt]
 {\eg} - \vgg= {\ftd} \quad &\mbox{ on } \; \partial D,
\\[6pt]
\partial ({\eg} - \vgg)/\partial \nu = {\ftn} \quad &\mbox{ on } \; \partial D,
\end{array}\right.
\end{equation}
for given $({\ftd},\ftn) \in H^{3/2}(\partial D) \times H^{1/2}(\partial D)$
where $\nu$ denotes the outward normal on $\partial D$. $k$ is called  a {\it transmission eigenvalue} if the homogeneous problem  (\ref{oitp}), i.e. with ${\ftd}=0$ and ${\ftn}=0$, has non-trivial solutions. Up-to-date results on this problem can be found in \cite[Chapter 3]{CCH} where in particular one finds sufficient solvability conditions. In the sequel we make the following assumption. \mfied{If the boundary of $D$  intersects the boundary of $\Omega_0$, then the previous interior transmission problem should be augmented with periodicity conditions on $\partial D \cap \partial \Omega_M$. Since this condition does not affect the assumptions on the solvability of the interior transmission problem (in $H^2(D)$ with periodic conditions on $\partial D \cap \partial \Omega_M$) nor requires any substantial modification of the arguments below (other than changing the solution space), we make the choice of simplifying this technicality and assume that $\partial D \cap \partial \Omega_0 =\emptyset$}.
\begin{asp} \label{HypoLSM}
\mfied{$\partial D \cap \partial \Omega_0 =\emptyset$} and the refractive index $n$ and \mfied{the wave number} $k>0$ are such that \eqref{oitp}
has a unique solution.
\end{asp}
\noindent
In particular,  if $\Re(n-1)>0$ or $-1<\Re(n-1)<0$ uniformly in a neighborhood of $\partial D$  inside $D$  the interior transmission problem (\ref{oitp}) satisfies the Fredholm alternative, and the  set of real standard transmission eigenvalues is discrete (possibly empty).  Thus Assumption \ref{HypoLSM} holds  as long as $k>0$ is not a transmission eigenvalue.

\noindent
 From now on, for $\zg \in \omehm$, we denote by $\sepmgre{z}: = \{ \coefpmgre{z}{\ell}\}_{\ell \in \Zd}$  the Rayleigh sequences of $\Phi(n_p,z)$ with $n_p=1$  define in (\ref{phi}) given by 
\begin{equation}\label{lunch}
	\coefpmgre{z}{\ell}:= \textstyle{\frac{\i}{2\normM{ML} \beta_\sper(\ell)}} e^{-\i (\alpm{\ell} \ol{\zg} - \betm{\ell}|\zg_d \mp h|)} .
\end{equation}
\begin{lemma}[Lemma 3.3 in  \cite{Thi-Phong3}]  \label{lemHerg} The operator $\opeh^\pm$ is compact and
  injective. Let  $\seth_{\inc}^{\pm}(D)$ be  the closure of the range of $\opeh^\pm$ in $L^2(D)$. Then 
\begin{equation} \label{eq:forHinc}
\seth_{\inc}^{\pm}(D) = \seth_{\inc}(D) := \{\vgg \in L^2(D): \;\;\Delta \vgg
+ k^2 \vgg = 0 \mbox{ in } D\}.
\end{equation}
\end{lemma}

\begin{theorem}[Theorem 3.5 in  \cite{Thi-Phong3}] \label{TheoG}
Assume that Assumptions \ref{Ass:nk} and \ref{HypoLSM} hold. Then the
operator ${\opeg^{\pm}}:  \seth_{\inc}(D) \rightarrow \ell^2(\Zd)$ defined by (\ref{defG}) is injective with dense range. Moreover $\sepmgre{z}$
belongs to $\Rang(\opeg^{\pm})$ if and only if $\zg \in D$.
\end{theorem}
\noindent
Another main ingredient is a symmetric factorization of  an appropriate operator given in terms of $\open^{\pm}$. To this end, for a generic operator $F:H\to H$, where $H$ is a Hilbert space, with adjoint $F^*$ we define
\begin{equation}\label{sharp}
 	\opef_{\sharp} := |\Re(\opef) |+ |\Im(\opef) |
 \end{equation}
 where \; $\Re(\opef) := \frac{1}{2}
\left(\opef +\Odual{\opef}\right)$, \; $\Im(\opef) := \frac{1}{2\i}
\left(\opef - \Odual{\opef}\right)$. 

\noindent
Now if $\opet: L^2(D) \rightarrow L^2(D)$ is defined by 
\begin{equation}  \label{defT}
	\opet \fg:= k^2(n - 1) (\fg + \wgg|_D)
\end{equation}
with $\wgg$ being the solution of \eqref{eq:w}, we have the following result under Assumptions \ref{Ass:nk}, \ref{HypoLSM}.
\begin{theorem}[Theorem 4.2 in  \cite{Thi-Phong3}] \label{TheoFactorization}
	The following factorization holds
	\begin{equation} \label{titifactN}
		\open_{\sharp}^{\pm} = \Odual{\opeh^{\pm}} \, \opet_{\sharp} \, \opeh^\pm,
	\end{equation} where $\opet_{\sharp}: L^2(D) \to L^2(D)$ is self-adjoint and coercive on $\seth_{\inc}(D)$.  Moreover, $\zg \in D$ if and only if  \; $\sepmgre{z} \in \Rang\left((\open_{\sharp}^{\pm})^{1/2}\right)$.
\end{theorem}

\noindent
The above theorem provides a rigorous method to recover the support of $D$. However this is not satisfactory since the aim is find only the support of $\omega$ and trying to reconstruct everything may not be feasible due to possible complicated structure of the periodic media and even useless if $\omega\subset D_p$. Our goal is to derive an imaging method that resolves only $\omega$ without knowing or recovering $D_p$. This leads us to introducing next the sampling operator for a single Floquet-Bloch mode whose analysis will bring up a new interior transmission problem. 

\noindent
We end this section by introducing some more notations to be used in the sequel. 
\mfied{\begin{definition}
 A function $u$ is called quasi-periodic with parameter $\xi = (\xi_1, \cdots ,\xi_{d - 1})$ and period $L = (L_1, \cdots , L_{d-1})$, with respect to the first $d-1$ variables (briefly denoted as $\xi-$quasi-periodic with period $L$) if:
\begin{equation*}
	u(\ol{x} + (jL), x_d) = e^{i\xi \cdot (jL)} u(\ol{x}, x_d), \quad \forall j \in \Z^{d-1}.
\end{equation*} 
\end{definition}}
 Let $q$ be a fixed parameter in $\Zd_M$, we denote by $\greq{z}$ the solution to 
\begin{equation}\label{phiq}
\Delta \greq{z} + k^2 \greq{z} =  -\delta_{z}\qquad \mbox{in} \;\Omega_0
\end{equation}
 and is $\alpha_q$ quasi-periodic with period $L$ \mfied{with $\alpha_q :=  2\pi  q/ (ML)$}. The Rayleigh coefficients $\sepmgreq{z}$ of $\greq{z}$ are given by 
\begin{equation}\label{hhh}
	\hspace*{-0.5cm}\coefpmgreq{z}{j} =  \left \{
\begin{array}{cl} 
   \textstyle{\frac{\i}{2\normM{L} \betm{q + M\,\ell}}} e^{-\i (\alpm{q + M\,\ell} \ol{\zg} - \betm{q + M\,\ell}|\zg_d \mp h|)}  & \, \mbox{if} \, j = q + M \ell, \; \ell \in \Zd, \\[1.5ex]
 0 &\, \mbox{if} \, j \neq q + M \ell, \;  \ell \in \Zd.
\end{array}
 \right.
\end{equation}
Furthermore we assume that each period of $D_p$ is composed by $J \in \N$ disconnected components and the defect $\domp$  as well as the  components that contains or have non-empty intersection with $\domp$ are in one cell, which we denote by $\Omega_0$ (otherwise we could rearrange the period). For convenience, we now introduce some additional notations.  We denote by $\Dpint$ the union of the components of $D_p \cap \Omega_0$ that have nonempty intersection with $\domp$,  and by $\Dpout$ its complement in $D_p\cap\Omega_0$, i.e the union of all the components of $D_p \cap \Omega_0$ that do not intersect $\omega$. Furthermore, we denote by $\Dint: = \Dpint \cup \domp$ and by $\Dtot : = \Dint \cup \Dpout$. Obviously,  $\Dtot = D \cap \Omega_0$. (see also Fig. \ref{Fig:ML-periodic} and note that if $\domp$ does not intersect with $D_p$ then $\Dpint \equiv \emptyset$, $\Dpout \equiv D_p \cap \Omega_0$ and $\Dint = \domp$). We consider the following $ML$-periodic copies of the aforementioned regions
\begin{equation}
\label{nota}	
 \Dpoutp =  \bigcup_{m \in \Z_M} \Dpout + mL,   \quad  \Dint_p := \bigcup_{m \in \Z_M} \Dint + mL  \quad \text{and}  \quad  \Dtot_p := \bigcup_{m \in \Z_M} \Dtot + mL
\end{equation} Remark that $\Dtot_p \equiv D_p \cup \big(\cup_{m \in \Z_M} \domp + mL\big)$ contains $D$ and the $L$-periodic copies of $\domp \setminus D_p$. {We remark that $n = n_p = 1$ in $\Dtot_p \setminus D$.}

\section{The Near Field Operator for a Single Floquet-Bloch Mode}

Let $\aher \in \ell^2(\Zd)$, we define for $q \in
\Zd_M$, the element $a_q \in \ell^2(\Zd)$ by
\[ 
	\aher_q(j) := \aher(q + jM).
\]
We then define the operator $\opei_q : \ell^2(\Zd) \to \ell^2(\Zd)$, which transforms $\aher \in \ell^2(\Zd)$ to $\tilde \aher \in \ell^2(\Zd)$ such that 
\begin{equation}\label{iop}
\tilde{\aher}_q = \aher\quad \mbox{and} \quad \tilde{\aher}_{q'} = 0 \; \mbox{ if } \; q \neq q'. 
\end{equation}
We remark that $\opei^{*}_q (\aher) = \aher_q$, where $\opei^{*}_q: \ell^2(\Zd)
\to \ell^2(\Zd)$ is the dual of the operator $\opei_q$. 
The single Floquet-Bloch mode Herglotz operator $\opeh^{\pm}_q: \ell^2(\Zd) \to L^2(D)$
is defined by 
\begin{equation} \label{2defHq}
	\opeh^{\pm}_q \aher := \opeh^{\pm} \opei_q \aher = \sum_{j} \aher(j) \uipm(\cdot; q + jM)|_{D}
\end{equation}
and the  single Floquet-Bloch mode  near field (or data) operator $\open^{\pm}_q:
\ell^2(\Zd) \to \ell^2(\Zd)$  is defined by 
\begin{equation} \label{2defNq} 
	\open^{\pm}_q \,\aher  = \opei^{*}_{q} \, \open^{\pm} \,  \opei_q \, \aher.
\end{equation} 
We remark that $\opeh^{\pm}_q \aher$ is an $\alpha_q-$quasi-periodic function
with period $L$. The sequence $\open^{\pm}_q \,\aher $ corresponds to the
Fourier coefficients of the $\alpha_q-$quasi-periodic component of the
scattered field in the decomposition \eqref{for:decompw}. This operator is then somehow
associated with $\alpha_q-$quasi-periodicity. One immediately sees from the factorization $\open^{\pm} = (\opeh^{\pm})^{*} \,  \opet \,
\opeh^{\pm}$ that the following factorization holds.
\begin{equation} \label{2FactfNq} 
	\open^{\pm}_q = (\opeh^{\pm}_q)^{*} \, \opet \, \opeh^{\pm}_q.
\end{equation}
For later use we also define the operator  $\opeg^{\pm}_q: \ol{\Rang(\opeh^{\pm}_q)} \to \ell^2(\Zd)$  by 
\begin{equation} \label{2defGq} 
	 \opeg^{\pm}_q = (\opeh^{\pm}_q)^{*} \opet |_{\ol{\Rang(\opeh^{\pm}_q)} }
\end{equation} where the operator $\opet$ is defined by \eqref{defT}. 

\noindent 
Observing that
\[
	\varphi(j;\ol{x}): =e^{\i \alpm{j} \ol{x} }  = e^{\frac{2\pi}{ML}j \ol{x}}, \qquad j\in \Z
\] is a Fourier basic of $ML$ periodic function in $L^2(\Omega_M)$, we have that any $w \in L^2(\Omega_M)$ which is $ML$ periodic, has the expansion 
\begin{equation}
	w(x) = \sum_{j \in \Z}\widehat{w}(j,x_d) \varphi(j;\ol{x}), \qquad \text{where} \qquad \widehat{w}(j,x_d) := \frac{1}{\normM{ML}} \int_{\Omega_M} w(x) \ol{\varphi(j;\ol{x})} \d{\ol{x}}.
\end{equation} Spliting $j$ by module $M$ we can arrange the expansion of $w$ as 
\begin{equation}
	w(x) = \sum_{q \in \Z_M} \Big(\sum_{\ell \in \Z} \widehat{w}(q+M\ell,x_d) \varphi(q+M\ell; \ol{x}) \Big),
\end{equation}
where $ \varphi(q+M\ell; \ol{x})$ is $\alpha_q-$quasi-periodic with period $L$, here $\alpha_q: = \frac{2\pi}{L}q$. Letting
\[
	w_q: = \sum_{\ell \in \Z} \widehat{w}(q+M\ell,x_d) \varphi(q+M\ell; \ol{x})
\] we have that  $w_q$ is $\alpha_q-$quasi-periodic with period $L$. Thus any $ML-$periodic  function $w \in L^2(\Omega_M)$ can be decomposed 
\begin{equation}
\label{for:decompw}
w = \sum_{q\in \Z_M} w_q
\end{equation} where $w_q$ is $\alpha_q-$quasi-periodic with period $L$. Moreover, by the orthogonality of the Fourier basic $\{ \varphi(j; \cdot) \}_{j \in \Z}$, we have that 
\begin{equation}
\raycoefpm{w_q}{j} = 0 \quad \text{if} \; j \neq q+ M\ell, \; \ell \in \Z \qquad \text{and} \qquad \raycoefpm{w}{q+M\ell} = \raycoefpm{w_q}{q+M\ell}
\end{equation} where $\raycoefpm{w_q}{j}$ the Rayleigh sequence of $w_q$ defined in \eqref{RayleighCh2}. By definition of $\opeg^{\pm}_q$, we see that $\opeg^{\pm}_q(f)$ is a Rayleigh sequence of $\raycoefpm{w}{j}$ at all indices $j = q + M\ell, \ell \in \Z$, where $w$ is solution of \eqref{eq:w}. Seeing also the line above that theses coefficients come from the Rayleigh sequence of $w_q$ where $w_q$ is one of the component of $w$ using the decomposition \eqref{for:decompw}, which is $\alpha_q-$ quasi periodic.   We now assume that $f|_{D_p}$ is $\alpha_q-$quasi-periodic. Then,  using the decomposition \eqref{for:decompw} for $w$, and that fact that $n_p$ is periodic, $f$ is $\alpha_q-$quasi-periodic and $n - n_p$ is compactly supported in one period $\Omega_0$, \eqref{eq:w} becomes
\begin{equation} 
	\Delta w_q + k^2n_p w_q  = k^2(n_p - n) w + k^2(1-n)f \qquad \text{in} \quad \Omega_0. 
\end{equation} Denoting by $\widetilde{w} : = w - w_q$, the previous equation is equivalent to 
\begin{equation} 
\label{eq:wq}
	\Delta w_q + k^2n w_q  = k^2(n_p - n) \widetilde{w} + k^2(1-n)f \qquad \text{in} \quad \Omega_0.
\end{equation} 
Therefore, operator $\opeg_q^{\pm}:\ol{\Rang(\opeh^{\pm}_q)} \rightarrow \ell^2(\Zd)$ can be equivalently defined as
\begin{equation} \label{defGqequiv}
\opeg^{\pm}_q (f) := \opei^{*}_q\{\raycoefpm{w_q}{\ell}\}_{\ell \in\Zd},
\end{equation} where $w_q$ solution of \eqref{eq:wq} and $w_q + \widetilde{w}$ is solution of \eqref{eq:w}. This definition is helpful for proving the  following properties for  $\opeh^{\pm}_q$ and $\opeg^{\pm}_q$ that are the counterpart results to  Lemma \ref{lemHerg} and  Lemma \ref{TheoG}, now needed for the operator $\open^{\pm}_q$.
\begin{lemma}
\label{lemHergq}
The operator $\opeh^{\pm}_q$ is injective and 
$$\ol{\Rang(\opeh^{\pm}_q)} = \seth^q_{\inc}(D) := \{v \in L^2(D), \quad   \Delta v + k^2 v = 0 \; \mbox{in } D \;  \mbox{and} \; v|_{D_p} \; \mbox{is $\alpha_q-$quasi-periodic}\}.$$
\end{lemma}
\begin{proof}
The proof of this lemma follows the lines of the proof of Lemma 5.1 in  \cite{Thi-Phong3} slightly modified  to account for more general location of the defect. $\opeh^{\pm}_q$ is injective since $\opeh^{\pm}$ is injective and $\opei_q$ is injective. We now prove that $(\opeh^{\pm}_q)^{*}$ is injective on $\seth^q_{\inc}(D)$. Let $\varphi \in \seth^q_{\inc}(D)$ and assume $(\opeh^{\pm}_q)^{*} (\varphi) = 0$. We define
\[
		u(x): = \frac{1}{\normM{M}}\int_{D} \Phi_q(x - y) \varphi(y)\d{y}. 
\] where $\Phi_q(x)$ has expansion
\begin{equation}
\label{form:phiq}
 \Phi_q(x) =  \frac{\i}{2\normM{L}}  \sum_{\ell \in \Zd} \frac{1}{\betm{q + M\ell}} e^{\i \alpm{q + M\ell} \ol{x} + \i\betm{q + M\ell} |x_d|}.
\end{equation} with $\alpm{q + M\ell} = \frac{2\pi}{ML}(q+M\ell)$ and $ \betm{q + M\ell} = \sqrt{k^2 - \alpm{q + M\ell}^2}$. By definition of $u$ and the expansion of $\Phi_q$ we have that 
\begin{equation}
\label{eq:coef:j}
\begin{array}{ll}
	\raycoefp{\eg}{j} &= \displaystyle  \frac{1}{|\interd{\boul}{\bour}|} \int_{\interd{\boul}{\bour}} \Big[ \frac{1}{\normM{M}}\int_{D} \Phi_q((\ol{x},h) - y) \varphi(y) \d{y} \Big] e^{ - \i \alpm{j} \cdot \ol{x}} \d{\ol{x}} \\[5ex] 
	 &=  \left\{\begin{array}{cl}
	 	\displaystyle  \int_{D} \varphi(y) \frac{\i }{2\betm{j}} e^{\i \alpm{j} - \i \betm{j}(x_d - h)} =  \left(\Odual{\opeh^{+}}(\varphi)\right)(j), & \text{if} \; j = q+M\ell \\[3ex]
	 	0 & \text{if} \; j \neq q + M \ell
	 \end{array}  \right. 
\end{array}
\end{equation}
which implies that $ \raycoefpm{u}{j} = 0$ for all  $j \neq q + M\ell$ and $ \raycoefpm{\eg}{q+M\ell} = \left(\Odual{\opeh^{\pm}}(\varphi)\right)(q+M\ell) =((\opeh^{\pm}_q)^{*} (\varphi))(\ell) = 0$. Therefore $u$ has all Rayleigh coefficients equal 0, which implies that 
\[ 
	u = 0, \quad \; \mbox{for} \; \pm x_d > h.
\]
We now observe that  for all $y \in D$, $\Delta \Phi_q(\cdot ; y) +
k^2 \Phi_q(\cdot ; y) = 0$ in the complement of $ \Dtot_p$. This implies that
$$\Delta u + k^2 u = 0 \quad \text{in } \; \R^d \setminus \Dtot_p. $$ 
Using a unique continuation argument we infer that $ u = 0$ in $\Omega_M
\setminus \Dtot_p$. Therefore, $u \in H^2_0(\Dtot_p)$ by the
regularity of volume potentials. We now consider two cases:

\bigskip

\noindent{\it If \; $\domp \subset D_p$}, then $\Dtot_p \equiv D_p$, i.e $u \in H^2_0(D_p)$. Moreover, by definition, $u$ verifies $\Delta u + k^2 u = -\varphi$ in $D_p$ and using the fact that $\Delta \varphi + k^2 \varphi = 0$ in $ D_p$ we finally have
\begin{equation}
-\|\varphi\|^2_{L^2(D_p)} = \int_{D_p} (\Delta u + k^2 u) \ol{\varphi} \d{x} = \int_{D_p} u (\Delta \ol{\varphi} + k^2 \ol{\varphi}) \d{x} = 0.
\end{equation}
This proves that $\varphi = 0$ and yields the injectivity of $(\opeh^{\pm}_q)^{*}$ on $\seth^q_{\inc}(D)$.

\medskip

\noindent{\it If \; $\domp \not \subset D_p$}, let denote by  $\cdomp: = \domp \setminus D_p$ then $\cdomp \neq \emptyset$. We then rewrite $u(x)$ as 
\[
		u(x): = \frac{1}{\normM{M}}\int_{\cdomp} \Phi_q(x - y) \varphi(y)\d{y} +
                \frac{1}{\normM{M}} \int_{D_p } \Phi_q(x - y) \varphi(y)\d{y}.
\]  Since $\varphi|_{D_p}$ and $\Phi_q$ are $\alpha_q-$quasi-periodic
functions with period $L$, then for $x \in D_p$ we have 
\[ 
	\int_{D_p \cap \Omega_{m}} \Phi_q(x - y) \varphi(y)\d{y} = \int_{D_p \cap \Omega_0} \Phi_q(x - y) \varphi(y)\d{y}, \quad \; \forall \; m \in \Z^{d-1}_M.
\]   Therefore, for $x \in D_p \cap \Omega_m$.
\begin{equation}
\label{eq:2DpOm}
	u(x): = \frac{1}{\normM{M}}\int_{\cdomp} \Phi_q(x - y) \varphi(y)\d{y} +
                \int_{D_p \cap \Omega_m} \Phi_q(x - y) \varphi(y)\d{y} 
\end{equation}
We recall that $\Delta \Phi_q(x -  \cdot) + k^2 \Phi_q(x - \cdot) = -\delta_x$ in
$\Omega_m$ and $\Delta \Phi_q(x - \cdot) + k^2 \Phi_q(x -\cdot) = 0$ in
$\cdomp$ (by quasi-periodicity of $\Phi_q(x-\cdot)$). Hence, from \eqref{eq:2DpOm} we obtain that for $m\in \Zd_M$,
\begin{equation}\label{2fintas1}
\Delta u(x) + k^2 u(x) =  - \varphi(x) \mbox{ in } D_p \cap \Omega_m. 
\end{equation}
Let us set for  $m\in \Zd_M$
$$ \varphi_m(x) := e^{\i \alpha_q \cdot mL} \varphi(x - mL) \mbox{  for  } x \in
\cdomp + mL.$$
Then we have, using the  $\alpha_q-$quasi-periodicity of $\Phi_q$ that  for $ x \in \cdomp+ mL$
\[
		u(x): = \frac{1}{\normM{M}}\int_{\cdomp+ mL} \Phi_q(x - y) \varphi_m(y)\d{y} +\frac{1}{\normM{M}}
                \int_{D_p} \Phi_q(x - y) \varphi(y)\d{y} 
\] where $\Delta \Phi_q(x - \cdot ) + k^2 \Phi_q(x - \cdot) = - 0$ in
$D_p$ and  $\Delta \Phi_q(x - \cdot) + k^2 \Phi_q(x - \cdot) = -\delta_x$ in
$\cdomp + mL$. We then get 
\begin{equation}\label{2fintas2}
	\Delta u(x) + k^2 u(x) = -  \varphi_m  \mbox{ in } \cdomp + mL.
\end{equation}
Now define the function $\widetilde \varphi $ by 
$$
\widetilde \varphi = \varphi \mbox{ in } D_p \mbox{ and } \widetilde \varphi =
\varphi_m \mbox{ in }  \cdomp + mL  \quad \mbox{
   with } \quad m\in \Zd_M.$$
 Clearly
 $$\Delta \widetilde \varphi + k^2 \widetilde\varphi = 0 \mbox{ in } \Dtot_p$$
Since  $u \in H^2_0(\Dtot_p)$ we then have 
\[
	\int_{D_p} (\Delta u + k^2 u) \ol{\widetilde \varphi}  =0.
\] This implies according to \eqref{2fintas1} and \eqref{2fintas2} that
\[
		\int_{D_p}|\varphi|^2 \d{x} +  \normM{M} \int_{\cdomp}|\varphi|^2 = 0,  
\]
which implies $\varphi = 0$ \; in \; $D$. This proves the injectivity of $\Odual{\opeh^{\pm}}$ on $\seth^q_{\inc}(D)$ and hence proves the Lemma.
\end{proof}

\noindent
The following theorem concerning the properties $\opeg^{\pm}_q$ requires the solvability of a new interior transmission problem (to be formulated later), which  up to this study was an open problem except for the case when $\omega\cap D_p=\emptyset$ investigated in \cite{tpnguyen}. 

\begin{asp} \label{as-nitp} The refractive index $n$ and $k>0$   are such that  the new interior transmission problem defined in Definition \ref{nitp} has a unique solution.
\end{asp}

\noindent
Section \ref{snitp} is dedicated to derive sufficient conditions for which Assumption \ref{as-nitp} holds true. 
\begin{theorem}
\label{Lem:injectiveG_q} Suppose that Assumptions \ref{Ass:nk}, \ref{HypoLSM} and \ref{as-nitp} hold. Then the operator $\opeg^{\pm}_q : \seth^q_{\inc}(D) \to \ell^2(\Zd)$ is injective with dense range. 
\end{theorem}
\begin{proof}
To prove the injectivity of $G_q$ we assume that $f \in \seth^q_{\inc}(D)$ such that $\opeg_q(f) = 0$. Let $w$ be solution of \eqref{eq:w} with data $f$. 
Since $f|_{D_p}$ is $\alpha_q-$quasi-periodic and $n_p - n$ is compactly supported of period $\Omega_0$ we then have from \eqref{defGqequiv} that 
\begin{equation*} 
\opeg^{\pm}_q (f) := \opei^{*}_q\{\raycoefpm{w_q}{\ell}\}_{\ell \in\Zd},
\end{equation*} where $\{\raycoefpm{w_q}{\ell}\}_{\ell \in\Zd}$ is the Rayleigh sequence of $w_q$ and $w_q$ is solution to
\begin{equation} \label{eq:2wq1}
	\Delta w_q + k^2n w_q  = k^2(n_p - n) (w-w_q) + k^2(1-n)f \qquad \text{in} \quad \Omega_0 
\end{equation} 
In particular we have that $\Delta w_q + k^2 w_q = 0$ in $\Omega_M \setminus \Dtot_p$ (we recall notations in (\ref{nota})).
Using a similar unique continuation argument as at the beginning of the
proof of Lemma \ref{lemHergq} we deduce that
\[
	w_q = 0 \quad \mbox{in} \quad \Omega_M \setminus \Dtot_p.
\] Actually, if $\domp \subset D_p$ then $\Dtot_p \equiv D_p$ thus  $f|_{\Dtot_p}$ is $\alpha_q-$quasi-periodic. However, if $\domp \not\subset D_p$, $f|_{\Dtot_p}$ is not $\alpha_q-$quasi-periodic. To restore  $\alpha_q-$quasi-periodicity, we introduce an other function $\fc$ by  
\begin{equation}
\label{def:ftilde}
	\fc: =  \left\{\begin{array}{ll}
		f &\quad \text{in} \quad \Omega_M \setminus \Dint_p  \\ [1.5ex]
		e^{i \alpha_q mL} f|_{\Dint} &\quad \text{in} \quad \Dpint + mL, \quad \forall \; m\in \Z_M.
\end{array}	
\right.	
\end{equation} i.e we keep  $f$  the same outside $\Dint_p$ and extend the values of $f$ in $\Dint$ by $\alpha_q$-quasi-periodicity to $\Dint_p$. Since $f|_{D_p}$ is $\alpha_q$-quasi-periodic then  the definition of $\fc$ implies that $\fc = f$ also  in $\Dint_p \cap D_p$. Therefore $\fc = f$ in $D$ and $\fc = f$ in $\Omega_M \setminus \Dtot_p$ (in other words $\fc \neq f$ in $\Dtot_p \setminus \ol{D}$) and $\fc|_{\Dtot_p}$ is $\alpha_q$-quasi-periodic with period $L$.  We remark that in the special case  when $\domp \subset D_p$, $\fc \equiv f$ and $\Dtot_p \equiv D_p$.  We now write \eqref{eq:2wq1}  in terms of $\fc$ (since $\fc \equiv f$ in $D$) 
\begin{equation} \label{eq:2w1.1}
	\Delta w_q + k^2n w_q  = k^2(n_p - n) (w-w_q) + k^2(1-n) \fc \qquad \text{in} \quad \Omega_0, 
\end{equation} hence it is enough to prove that $\fc = 0$ in $\Dtot_p$. To this end, using the fact that $\fc = f$ in $\Dtot$ (recall that $\Dtot = \mfied{\Dint \cup \Dpout = } \Dtot_p \cap \Omega_0 = D \cap \Omega_0)$  then $\fc$ verifies 
\[
	\Delta \fc + k^2 \fc = 0 \qquad \text{in} \quad \Dtot_p
\] and by the $\alpha_q$-quasi-periodicity of $w_q$ and $\fc$, it is sufficient to prove that $\fc = 0$ in the cell $\Omega_0$, i.e,  proving that the following problem 
\begin{equation}
\label{eq:itpGq}
\left \{
	\begin{array}{ll}
		\Delta w_q + k^2n w_q  =  k^2(n_p - n) (w-w_q) + k^2(1-n)\fc & \text{in} \quad  \Dtot,\\[1.5ex]
		\Delta \fc + k^2 \fc  = 0 & \text{in} \quad  \Dtot,
	\end{array}
\right.
\end{equation}
\mfied{has trivial solution $(w_q, \fc ) \in H_0^2(\Dtot) \times L^2(\Dtot)$}.
 Since in $\Dpout$, $n = n_p$ and $\Dpout \cap \Dint = \emptyset$, then  \mfied{$(w_q, f) \in w_q \in H_0^2(\Dpout)\times L^2(\Dpout)$} verifies
\begin{equation}
\label{eq:2itpwq1}
\left \{
	\begin{array}{lc}
		\Delta w_q + k^2n w_q  =  k^2(1-n)f & \text{in} \quad   \Dpout,\\[1.5ex]
		\Delta f + k^2 f  = 0 & \text{in} \quad   \Dpout.
		\end{array}
\right.
\end{equation}
 Assumption \ref{HypoLSM} implies that equation \eqref{eq:2itpwq1} has a trivial solution,  and therefore
 \[
	w_q  = f = 0 \quad \text{in} \quad \Dpout.
\] 
It remains  to prove that $f = 0$ in $ \Lambda$. We can now rewrite (\ref{eq:itpGq}) as a problem only in
 $\Dint$
\begin{equation}
\label{eqnew0}
\left|
	\begin{array}{lc}
		w_q \in H_0^2(\Dint) \mfied{ \mbox{ and } \fc \in L^2(\Dint)} \\[1.5ex]
		\left\{
		\begin{array}{ll}
		\Delta w_q + k^2n w_q  =  k^2(n_p - n) (w-w_q) + k^2(1-n)\fc & \text{in} \quad  \Dint,\\[1.5ex]
		\Delta \fc + k^2 \fc  = 0 & \text{in} \quad  \Dint.
		\end{array}
		\right.
	\end{array} 
\right.
\end{equation}
To deal with this problem, we first \mfied{express the quantity } $w - w_q $ in terms of $\fc$ using the property that $\fc = 0$ outside $\Dint$. To this end, recalling that  $\fc = f$ in $D$, we can  write \eqref{eq:w}  in terms of $\fc$ as
\begin{equation} \label{eq:2w1.1}
	\Delta w + k^2 n_p w = k^2 (n_p - n) w + k^2 (1 - n) \fc
\end{equation}
and then have
\begin{equation}
\label{eq:2forw1.1}
	w(x) = -\int_{D} \Big(k^2(n_p - n)w + k^2(1 - n) \fc \Big)(y) \fundnp{x-y} \d{y}.
\end{equation}
Using the facts that $\fc = 0$ and $n = n_p$ in  $\Dpoutp$, i.e. $n_p = n = 1$ in $\Dint_p \setminus D$ we have 
\begin{eqnarray}
\label{eq:reforw}
	w(x) &=& -  \int_{D \setminus \Dpoutp} \Big(k^2(n_p - n)w + k^2(1 - n) \fc \Big)(y) \fundnp{x-y} \d{y}. \nonumber \\
	&=& - \int_{\Dint_p} \Big(k^2(n_p - n)w + k^2(1 - n) \fc \Big)(y) \fundnp{x-y} \d{y}. \nonumber  \\
	&=& -k^2 \int_{\Dint} \Big((n_p - n)w +(1 - n) \fc \Big)(y) \fundnp{x-y} \d{y} \nonumber \\ 
	&-& k^2 \int_{\Dint_p \setminus \Dint} (1 - n_p) \fc (y) \fundnp{x-y} \d{y}.
\end{eqnarray}
Moreover, \mfied{since $w_q \in H^2_0(\Dint)$}, then
for all $\theta \in H^2(\Dint)$ satisfying $\Delta \theta + k^2 n_p \theta = 0$ we have
\begin{equation}
\int_{\Dint} \Big( \Delta w_q + k^2 n_p w_q\Big)	\mfied{\theta} \d x = 0,
\end{equation}
implying \mfied{from equation \eqref{eqnew0},} that
\begin{equation}
	\int_{\Dint} \Big( \mfied{(n_p - n)} w + (1- n) \fc \Big)	\mfied{\theta}\d x = 0.
\end{equation}
Remark that for $x \notin \Dint$, \mfied{$\Delta \fundnp{x-y} + k^2 n_p \fundnp{x-y} = 0$} for all $ y\in\Dint$. Combined with $\fc = 0$ outside $\Dint_p$, we then conclude from \eqref{eq:reforw} that 
\begin{equation}
	w(x)  = - \int_{\Dint_p \setminus \Dint} k^2(1 - n_p) \fc (y) \fundnp{x-y} \d{y} \quad \mbox{ for } x \notin \Dint. 
\end{equation}
Next we define 
\begin{equation}
\label{def:twpr}
	\twpr(x) =  - \int_{\Dint_p \setminus \Dint} k^2(1 - n_p) \fc(y) \fundnp{x-y} \d{y}, \quad \, x \in \Omega_M.
\end{equation}
then $\wpr - \twpr \in H^2_0(\Dint)$ and $ \Delta \twpr + k^2 n_p \twpr  = 0$ in $\Dint$. We now keep $w$ and $w_q$ as above and let  $\widehat{w}: = w_q + \twpr$ in $\Dint$ which obviously verifies 
\[
	\Delta \widehat{w} + k^2n \widehat{w}  =   k^2(1-n)\fc \quad \text{in} \; \Dint.
\] By Assumption \ref{Ass:nk} we have $w = \widehat{w}$ in $\Dint$. This proves that $\twpr = w - w_q$ in $\Dint$. 
 Moreover, using the $\alpha_q$ quasi-periodicity of $\fc$ in $\Dint_p$ and the periodicity of $n_p$ we have that 
\begin{align*}
	\int_{\Dint + mL}k^2(1 - n_p) \fc(y) \fundnp{x-y} \d{y} &=  \int_{\Dint} k^2(1 - n_p) \fc(y + mL) \fundnp{x -(y + mL)} \d{y}\\
	&= e^{\i \alpha_q m L} \int_{\Dint} k^2(1 - n_p) \fc(y) \fundnp{x -mL - y} \d{y}.
\end{align*}
Letting  for $y \in \Dint$
\begin{equation}
\label{formul:tphi}
	\tfundnp{x,y}: = \sum_{0 \neq m \in \Z_M} e^{\i \alpha_q m L} \fundnp{x-mL-y},
\end{equation}  we see that $\twpr$ defined by \eqref{def:twpr} is equivalent to 
\begin{equation}
	\twpr(x) =  - \int_{ \Dint} k^2(1 - n_p) \fc(y) \tfundnp{x,y} \d{y}.
\end{equation}
This leads us to define the operator $\opes{k}: L^2(\Dint) \to L^2(\Dint)$
\begin{equation}
\label{def:Skequiv}
\opes{k}:f \mapsto   - \int_{ \Dint} k^2(1 - n_p) f(y) \tfundnp{x,y} \d{y}.
\end{equation}
From the smoothing property of the volume potential we have that  the operator $\opes{k}$ is compact and $\opes{k}(f)$ satisfies
\[
 	\Delta \opes{k}(f) + k^2 n_p \opes{k}(f)  = 0  \quad \mbox{ in } \R^d \setminus  \Dint. 
\] 
From the above, we write $w-w_q$ in terms  $\opes{k}(f)$ then reformulate (\ref{eqnew0}) for $w_0:=w_q$ and $f$, we finally obtain
\begin{equation}
\label{eq:injitp}
	\left|
	\begin{array}{l}
	\wpr_0 \in  H^2_0(\Dint), \; f\in L^2(\Dint), \\[1.5ex]
\left\{
		\begin{array}{ll}
			
			\Delta \wpr_0 + k^2 n \wpr_0 = k^2(1 - n) f + k^2(n_p - n) \opes{k}(f) &\qquad \text{in} \quad \Dint,\\[1.5ex]
			\Delta f + k^2 f = 0 & \qquad \text{in} \quad \Dint.\\[1.5ex]
		\end{array} 
	\right.
		\end{array} 
	\right.
\end{equation} 
This problem is the homogeneous version of the new interior transmission problem defined  in Definition \ref{nitp} below, where $u:=w_0+f$.  Assumption \ref{as-nitp} now implies that $\wpr_0 = f = 0$ in $\Dint$, which proves the injectivity of $G_q$.
\end{proof}
\begin{definition}[The new interior transmission problem]\label{nitp}{\em Find $(u, f) \in
L^2(\Dint) \times L^2(\Dint)$ such that $
u -f \in H^2(\Dint) $
 and 
\begin{equation} \label{eqnew}
\left\{ \begin{array}{lll}
\Delta {\eg} + k^2 n {\eg} = k^2(n_p - n) \opes{k}(f) \quad& \mbox{ in } \; \Dint, 
\\[6pt]
\Delta f + k^2 f = 0  \quad &\mbox{ in } \; \Dint,
\\[6pt]
 {\eg} - f= {\ftd} \quad &\mbox{ on } \; \partial \Dint,
\\[6pt]
\partial ({\eg} - f)/\partial \nu = {\ftn} \quad &\mbox{ on } \; \partial \Dint, 
\end{array}\right.
\end{equation}
for given $({\ftd},\ftn) \in H^{3/2}(\partial \Dint) \times H^{1/2}(\partial \Dint) $
where
\begin{align}
\label{def:2Stilde}
\opes{k}: L^2( \Dint)& \to L^2(\Dint):  \nonumber \\
f &\mapsto  - \int_{\Dint} k^2(1 - n_p) f(y) \Big(\sum_{0 \neq m \in \Z_M} e^{\i \alpha_q m L} \fundnp{x-mL-y}\Big) \d{y},
\end{align} $\fundnp{\cdot}$ is the $ML$-periodic fundamental solution given by (\ref{phi}), and  $\nu$ denotes the unit normal on  $\partial \Dint$ outward to $\Dint$. 
}
\end{definition}

\noindent
\begin{definition}\label{ntp-def}
{\em Values of $k\in {\mathbb C}$ for which the homogenous  problem with $\varphi=\psi=0$,  are called {\it new transmission eigenvalues}.
}
\end{definition}

\noindent
For sake of completeness, we end this section with proving a range statement for the operator $\opeg^{\pm}_q$ which is used in the imaging algorithm. To this end, we recall  $\Phi_q(\cdot;z)$ defined by (\ref{phiq}), $I_q$ given by (\ref{iop}) with adjoint $I^*_q$.
\begin{theorem}
\label{Lem:injectiveG_q2} Suppose that Assumptions \ref{Ass:nk}, \ref{HypoLSM} and \ref{as-nitp} hold. Then, $\opei^{*}_q \sepmgreq{z} \in \Rang(\opeg^{\pm}_q)$ if and only if $z \in \Dtot_p$.
\end{theorem}
\begin{proof}
 We consider two cases: \\
\noindent {\it Case 1: $z \in \Dtot_p = \Dint_p \cup \Dpoutp$}. 

\medskip

{\it(i) If $z \in \Dpoutp$:} Let  $(u, v) \in L^2(D) \times L^2(D)$ be the unique solution of (\ref{oitp}) with $\varphi: =\Phi_q(\cdot-z))|_{\partial D}$ and $\psi:=\partial   \Phi_q(\cdot-z))/\partial \nu|_{\partial D}$  and define 
$$w = \left\{ \begin{array}{cl} u - v
 \quad &\mbox{in} \quad \Dpoutp \\
 \Phi_q \quad &\mbox{in} \quad \Omega_M \setminus \Dpoutp.
\end{array} \right.$$ Then $w \in H^2_{\loc}(\Omega_M)$ and verifies equation \eqref{eq:w} with  $f = v$ in $\Dpoutp$ and $f = -\Phi_q$ in $\Omega_M \setminus \Dpoutp$. Therefore $\opeg^{\pm}(f) = \sepmgreq{z}$. Furthermore $u|_{\Dpoutp}$ and $v|_{\Dpoutp}$ are $\alpha_q$-quasi-periodic (due to the periodicity of domain $\Dpoutp$ and $\alpha_q$-quasi-periodicity of the data), hence $f$ is also $\alpha_q$ quasi-periodic. This implies $f \in \seth^q_{\inc}(D)$ and $\opeg_q^{\pm}(f) = \opei^{*}_q \opeg_q^{\pm}(f) =  \opei^{*}_q\sepmgreq{z}$. 

\medskip

{\it (ii) If $z \in \Dint_p$}: We {\it first} consider that $z \in \Dint = \Dint_p \cap \Omega_0$. Let  $(u, v) \in L^2(\Dint_p) \times L^2(\Dint_p)$ be the  $\alpha_q$-quasi-periodic extension of  $(u_{\Dint}, v_{\Dint})$, where  $(u_{\Dint}:=u, v_{\Dint}:=f)$ is the solution of  the new interior transmission problem in Definition \ref{nitp} with  $\varphi: =\Phi_q(\cdot;z))|_{\partial \Dint}$ and $\psi=:\partial   \Phi_q(\cdot;z))/\partial \nu|_{\partial \Dint}$. We then define
\[
	w_q = \left\{ \begin{array}{cl} u - v 
 \quad &\mbox{in} \quad \Dint_p \\
 \Phi_q \quad &\mbox{in} \quad \Omega_M \setminus\Dint_p.
\end{array} \right.
\]
Let $f := v$ in $\Dint_p$ and $f := -\Phi_q$ in $\Omega_M \setminus \Dint_p$ then $f \in \seth^q_{\inc}(D)$ and $w_q \in H^2_{\loc}(\Omega_M)$ satisfies the scattering
problem \eqref{eq:wq} with data $f$. Furthermore, $w$ defined such as $w:=w_q + \opes{k}(f)$ in $\Dint$ and $w: = w_q$ in $D \setminus \Dint$  is solution to \eqref{eq:w} with data $f$. Therefore $\opeg_q^{\pm}(f) =  \opei^{*}_q\sepmgreq{z}$.\\
 We {\it next} consider   $z \in \Dint + mL$ with $0 \neq m \in \Zd_M$, and  recall that $\sepmgreq{z}= e^{\i mL \cdot \alpha_q} \sepmgreq{z - mL}
$. If we take  $f\in \seth^q_{\inc}(D)$ such that $G_q^\pm(f) =
\opei^{*}_q \sepmgreq{z - mL}$, which is possible by the previous step since
$z - mL \in \Dint$, then 
\[
	 G_q^\pm(e^{\i mL \cdot \alpha_q} f) = \opei^{*}_q(\sepmgreq{z}).
\]

\noindent {\it Case 2: $z \notin \Dtot_p$}. If  $G_q^\pm(v) =
\opei^{*}_q \sepmgreq{z}$, then using the same unique continuation
argument as in  the proof of Lemma \ref{Lem:injectiveG_q}  we obtain $w_q =  \Phi_q $ in $\Omega_M \setminus \{\Dtot_p\}$ where $w_q$ is defined by
        \eqref{for:decompw} with $w$ being  the solution of \eqref{eq:w}  with $f=v$. This gives a contradiction since $w_q$ is locally
        $H^2$ in $\Omega_M \setminus \{\Dtot_p\})$ while $\greq{z}$ is
        not.
\end{proof}

\section{The Analysis of the New Interior Transmission Problem} \label{snitp}
This section is devoted to the study of the solvability of the new interior transmission problem in Definition \ref{nitp}. It  provides sufficient conditions on  $n$ and $k$ for which this problem is well-posed,  i.e. such that Assumption \ref{as-nitp} holds. As described in the previous section the  solvability of the new interior transmission problem is fundamental to ensuring the properties needed for the imaging of the defect $\omega$ with a single  Floquet-Bloch mode. Up to now the only case that could be handled was when $\omega\cap D=\emptyset$ \cite{tpnguyen} (see also  \cite{Thi-Phong3}). Here we provide a general analysis that cover all possible cases. Our approach generalizes  \cite{kirschitp} and \cite{Sylve2012}. 

\noindent
We start with proving  the following technical lemma:
\begin{lemma}
\label{lem:pro1}
Assume that \mfied{ $ n_p > \alpha > 0$ on $\R^d$}. Then there exists $ \theta >0$ and $C > 0$ such that 
	\[
		\| \opes{i\kappa}(f)\|_{L^2(\Dint)} \leq C e^{-\theta \kappa} \| f\|_{L^2(\Dint)}, \quad \kappa> 0
	\] for all $f \in L^2(\Dint)$. 
\end{lemma}
\begin{proof}
Denoting $\twpr: = \opes{i\kappa} (f)$ and $\fc$ the extension of  $f$ as  $\alpha_q-$quasi-periodic in $\Dint_p$, we have that 
\begin{equation*}
	\twpr(x) =  \kappa^2 \int_{\Dint_p \setminus \Dint} (1 - n_p ) \fc(y) \fundnp{x-y} \d{y}. 
\end{equation*}
where \mfied{$\fundnp{\cdot}$ denotes here the $ML-\text{periodic}$ fundamental solution associated with $k=i\kappa$. }
\mfied{An application of the Cauchy-Schwarz inequality an the Fubini theorem implies
\begin{align*}
	\| \twpr\|^2_{L^2(\Dint)} &\leq \kappa^4|\Dint_p\setminus \Dint|\int_{\Dint } \int_{\Dint_p \setminus \Dint} \big|( n_p-1 ) \fc(y) \fundnp{x-y}\big|^2 \d{y} \d{x} \\
	& = \kappa^4 |\Dint_p\setminus \Dint|\int_{\Dint_p \setminus \Dint} \big|( n_p-1 ) \fc(y) \big|^2 \int_{\Dint} \big| \fundnp{x-y} \big|^2\d{x} \d{y}.
\end{align*}
}
Next we let
\[
	\Sigma = \{z: = x - y, \quad x \in \Dint, \; y \in \Dint_p \setminus \Dint \}, 
\] 
\[
	d_{max} \in \R: \quad d_{max} > \sup \{|z|, z \in \Sigma \} \quad \text{and} \quad d: = d(\Dint, \Dint_p \setminus \Dint),
\]
and remark that $\forall x\in \Dint$, $\forall y\in \Dint_p \setminus \Dint$, $|x - y| >d$, hence
\begin{equation}
\label{def:bal}
\Sigma \subset \BB: = B(0, d_{max}) \setminus B(0,d)
\end{equation}
where $B(0,d)$ is a ball of radii $d$ and centered at the origin. \mfied{Remark that $d>0$ by Assumption 2}.   \mfied{We then have 
\begin{align*}
	\| \twpr\|^2_{L^2(\Dint)} &\leq \kappa^4|\Dint_p\setminus \Dint| 
	\int_{\Dint_p \setminus \Dint} \big|( n_p-1 ) \fc(y) \big|^2 \d{y}\int_{\BB} \big| \fundnp{z} \big|^2\d{z}.
\end{align*}
}
Since $\fc = f$ in $\Dint$, $\fc$ is quasi-periodic and $n_p$ is periodic in $\Dint_p$, then 
$$
\int_{\Dint_p \setminus \Dint} \big|( n_p-1 ) \fc(y) \big|^2 \d{y} = (|M|-1)\int_{\Dint} \big|( n_p-1 ) \fc(y) \big|^2 \d{y} \leq (|M|-1) \sup_{\Dint}|1 - n_p|\|f \|^2_{L^2(\Dint)}.
$$
 We now prove that, there exists $\theta > 0$ and $C_0 >0$ such that 
\begin{equation}
\label{proof:tphi}
	\int_{\BB} \big| \fundnp{z} \big|^2\d{z} \leq C_0 e^{-\theta \kappa}.
\end{equation}
We recall that $\fundnp{z}$ \mfied{is $ML-$ periodic and satisfies 
 \begin{equation}
	 (\Delta - \kappa^2 n_p) \fundnp{\cdot} = - \delta_{0} \quad \text{in} \; \Omega_M\end{equation}
and  $\fundnp{\cdot} \in L^2(\Omega_M)$ (or equivalently the Rayleigh radiation condition (\ref{C4:RDC}) with $k=i\kappa$).}
\mfied{Consider the  fundamental solution  $\Psi \in L^2(\R^d)$  satisfying
$$(\Delta - \kappa^2 n_p)\Psi = - \delta_{0} \qquad \mbox{in } {\mathbb R}^d.$$
Since $n_p$ is positive definite, then one can prove that $e^{\gamma \kappa |x|} \Psi(x) \in L^2(\R^d)$ for $\gamma >0$ sufficiently small (following the lines of the proof of Theorem 4.4 in \cite{Thi-Phong2}).
The function $u^s: = \fundnp{\cdot } - \Psi$ verifies
$$
\Delta u^s - \kappa^2 n_p u^s = 0  \quad \text{in} \; \Omega_M
$$
and application of the Green formula and using the periodicity conditions of $\fundnp{\cdot }$ imply
$$
\int_{\Omega_M} \Big( |\nabla u^s |^2 + \kappa^2 n_p |u^s|^2\Big) \d{x} \le C \left( \| \Psi \|_{H^{\frac{1}{2}}{(\partial \Omega_M)}}  \| \frac{\partial u^s}{\partial \nu} \|_{H^{-\frac{1}{2}}{(\partial \Omega_M)}}  + \| u^s \|_{H^{\frac{1}{2}}{(\partial \Omega_M)}}  \| \frac{\partial \Psi}{\partial \nu}\|_{H^{-\frac{1}{2}}{(\partial \Omega_M)}} \right)
$$
for some constant $C$ independent from $\kappa$. The decay property of $\Psi$ implies that  
$$
\left( \| \Psi \|_{H^{\frac{1}{2}}{(\partial \Omega_M)}} + \| \frac{\partial \Psi}{\partial \nu}\|_{H^{-\frac{1}{2}}{(\partial \Omega_M)}} \right) \le e^{-\gamma_1 \kappa }
$$
for $\kappa$ sufficiently  large and $\gamma_1 >0$ a constant independent from $\kappa$. Therefore, by classical continuity properties for  traces and normal traces and the fact that $\Delta u^s = \kappa^2 n_p u^s$, one conclude that for $\kappa$ sufficiently  large
$$
\int_{\Omega_M} (|\nabla u^s |^2 +  | u^s |^2) dx \le C e^{-\gamma_1 \kappa }
$$
for some constant $C$ independent from $\kappa$. One then obtain the desired estimate \eqref{proof:tphi} by writing   $\fundnp{\cdot } = \Psi + u^s$ and combining the previous estimate with the exponential decay of $\Psi$.}
\end{proof}

\noindent
We now turn our attention to the analysis of the new interior transmission problem in Definition \ref{nitp}. For the given $(\varphi,\, \psi)\in H^{3/2}(\Dint)\times H^{1/2}(\Dint)$ in (\ref{eqnew}) we construct a lifting  function $u_0\in H^2(\Dint)$ such that $u_0|_{\partial \Dint}=\varphi$ and $\partial u_0/\partial \nu|_{\partial \Dint}=\psi$. Hence $w:=u-u_0-f\in H^2_0(\Dint)$ and we let $F:= (\Delta u_0+k^2nu_0)/k^2 \in L^2(\Dint)$.  To further simplify notation, we set $\lambda: = -k^2\in{\mathbb C}$, $q: =  n-1$, $p: = n - n_p$ and  $v: = -\wpr /k^2 \in H^2_0(\Dint)$. In these notations, the problem we need to solve reads: 
\begin{equation}
\label{eq:2.2newitp}
\left|
\begin{array}{l}
v \in  H^2_0(\Dint) \mbox{ and }f \in L^2(\Dint), \\[1.5ex]
	\left \{
		\begin{array}{ll}
			\Delta v - \lambda (q+1) v = q f + p \opes{\sqrt{-\lambda}}(f) +F & \quad \mbox{in } \Dint,\\[1.5ex]
			\Delta f  - \lambda f = 0 & \quad \mbox{in } \Dint, 
		\end{array}
	\right.
		\end{array}
	\right.
\end{equation}   
for a given $F\in L^2(\Dint)$, 	{where} $\opes{\sqrt{-\lambda}}(f)$ {is given by \eqref{def:2Stilde}}. We remark that \eqref{eq:2.2newitp} is a modification of the problem considered in \cite{kirschitp} (see also \cite[Section 3.1.3]{CCH}). We write this problem in an equivalent variational form. To this end, let us denote $\X: = H^2_0(\Dint) \times L^2(\Dint)$ with the norm $\|(v,f)\|^2_{\X} := \|v\|^2_{H^2(\Dint)} + \|f \|^2_{L^2(\Dint)}$ and the corresponding inner product $\langle \cdot, \cdot\rangle_{\X}$. Then we define the sesquilinear form $a_k: \X \times \X \to \CC$ by
\begin{equation}
	a_\lambda(v,f;\phi,\psi) = \int_{\Dint} (\Delta \ol{\phi}  - \lambda \ol{\phi}) f \d{x} + \int_{\Dint} \big(\Delta v - \lambda (q+1) v \big) \ol{\psi} - \big(q f + p \opes{\sqrt{-\lambda}}{f} \big) \ol{\psi} \d{x},
\end{equation} for $(v,f) \in \X$ and $(\phi,\psi) \in \X$ and the bounded linear operator $A_\lambda: \X \to \X$ defined by means of the Riesz's  representation theorem
\begin{equation}
	a_\lambda(v,f;\phi,\psi) = \langle A_\lambda(v,f),(\phi,\psi) \rangle_{\X}.
\end{equation}
Letting $\ell\in \X$ be the Riesz's representative of the conjugate linear functional
$$\big(\ell, (\phi,\psi)\big)_{\X}=\int_\Dint F \ol{\psi}  \d{x},$$
solving (\ref{eq:2.2newitp}) is equivalent to solving 
\begin{equation}\label{var1}
A_\lambda(v,f)=\ell \qquad  \mbox{or} \qquad a_\lambda(v,f;\phi,\psi)=\int_\Dint F \ol{\psi}  \d{x}, \qquad \forall \, (\phi,\psi)\in\X.
\end{equation}

\noindent
In a similar fashion, we define the sesquilinear form $b_k: \X \times \X \to \CC$ by 
\[
	b_\lambda(v,f;\phi,\psi) = \int_{\Dint} (\Delta \ol{\phi}  - \lambda \ol{\phi}) f \d{x} + \int_{\Dint} \big(\Delta v - \lambda v  - q f \big) \ol{\psi} \d{x}
\] with the associated bounded linear operator $B_\lambda: \X \to \X$ such that 
\[
	b_\lambda(v,f;\phi,\psi) = \langle B_\lambda(v,f),(\phi,\psi) \rangle_{\X},
\] 
\begin{lemma}
\label{lem:compAmA}
	For any $\lambda, \mu \in {\mathbb C}$, the differences $A_{\lambda} - B_{\mu}$  and  $A_{\lambda} - A_{\mu}$ is compact.
\end{lemma}

\begin{proof}
Let $(v_j,f_j) \in \X$  be an arbitrary sequence converging weakly  to $(0,0)$  in $\X$. We must show that $(A_\lambda - B_\mu) (v_j, f_j)$ converges  $(0,0)$ strongly in $\X$.  Recall that
$$
	\langle (A_\lambda - B_\mu) (v,f), (\phi,\psi) \rangle_{\X}
 = \big(a_\lambda - b_\mu\big) (v, f; \phi,\psi)
$$
where
$$
 \big(a_\lambda - b_\mu\big) (v_j, f_j; \phi,\psi) = (\mu - \lambda)\int_{\Dint} f_j \ol{\phi} +\int_{\Dint} \big(\mu-\lambda(q+1)\big)v_j \ol{\psi} - \int_{\Dint} p (\opes{\sqrt{-\lambda}} f_j - \opes{\sqrt{-\mu}} f_j) \ol{\psi}.
$$
 Let us take $(\phi, \psi) \in \X$ such that $\| (\phi, \psi) \|_\X = 1$ and  define $g_j \in H^1(\Dint)$ such that $\Delta g_j = f_j$ and $g_j = 0$ on $\partial \Dint$. Hence $g_j \rightharpoonup 0$ in $H^1(\Dint)$ implying that $g_j \to 0$ in $L^2(\Dint)$ by compact embedding of $H^1(\Dint)$ in $L^2(\Dint)$. Obviously,
$$
\Big| \int_{\Dint} f_j \ol{\phi} \Big| = \Big|\int_{\Dint} \Delta g_j \ol{\phi}\Big|  = \Big|\int_{\Dint} g_j \Delta \ol{\phi} \Big| \leq \|\Delta \ol{\phi} \|_{L^2(\Dint)}  \|g_j \|_{L^2(\Dint)} \leq \|g_j \|_{L^2(\Dint)}.
$$
Similarly, since $v_j \rightharpoonup 0$ in $H^2(\Dint)$, then $v_j \to 0$ in $L^2(\Dint)$, and 
$$
 	\Big| \int_{\Dint} \big(\mu-\lambda(q+1)\big) v_j \ol{\psi} \Big| \leq \sup_{\Dint}| \mu - \lambda(q+1)|  \|v_j \|_{L^2(\Dint)}  \|\psi_j \|_{L^2(\Dint)} \leq \sup_{\Dint}|  \mu-\lambda(q+1)|  \|v_j \|_{L^2(\Dint)}
$$ and 
 \begin{align*}
 	\Big|\int_{\Dint} p (\opes{\sqrt{-\lambda}} f_j - \opes{\sqrt{-\mu}} f_j) \ol{\psi} \Big|& \leq \sup_{\Dint}|p|  \big\|\opes{\sqrt{-\lambda}} f_j - \opes{\sqrt{-\mu}} f_j \big\|_{L^2(\Dint)}  \|\psi_j \|_{L^2(\Dint)} \\
 	&\leq  \sup_{\Dint}|p|  \big\|(\opes{\sqrt{-\lambda}} - \opes{\sqrt{-\mu}}) f_j \big\|_{L^2(\Dint)}.
 \end{align*}
 Since $(\opes{\sqrt{-\lambda}} - \opes{\sqrt{-\mu}}): L^2(\Dint) \to L^2(\Dint)$ is a compact linear operator,  then $ \|(\opes{\sqrt{-\lambda}} - \opes{\sqrt{-\mu}}) f_j \|_{L^2(\Dint)}$ converge to $0$ strongly in $L^2(\Dint)$.
 Therefore 
 \begin{align*}
 	\| (A_\lambda - B_\mu) (v_j, f_j) \|_{X \to X}  &= \sup_{\|(\phi,\psi) \|_{X} =1} |(a_\lambda - b_\mu) (v_j, f_j; \phi,\psi)| \\
 	&\leq c \Big( \|v_j \|_{L^2(\Dint)} + \|g_j\|_{L^2(\Dint)} + \big\|\opes{\lambda} f_j - \opes{\mu} f_j \big\|_{L^2(\Dint)} \Big) \longrightarrow 0.
 \end{align*}
The proof for $A_{\lambda} - A_{\mu}$ follows exactly the same lines.
\end{proof}

At this point we need to assume a sign condition on $q:=n-1$. To this end, let $R$ be a neighborhood of $\partial \Dint$ in $\Dint$ and denote by
\begin{equation}
q_{min}:=\inf\limits_{\Dint}{\Re(q)}>-1, \qquad q_{\star}:=\inf\limits_{R}{\Re(q)},\qquad q^{\star}:=\sup\limits_{R}{\Re(q)}.
\end{equation}

\noindent
The following lemmas, which are proven  first in \cite{kirschitp} for real refractive index and adapted to the case of complex refractive index in  \cite[Section 3.1.3]{CCH}),  play an important role in our analysis.
\begin{lemma}
\label{lem:Kirsch}
	Assume that $q\in L^\infty(\Dint)$ is such that $q_{min}+1>0$ and either  $q_\star>0$ or $q^\star<0$. Then there exists $c>0$ and $d>0$ such that for all $\lambda > 0$ the following estimates holds
\begin{equation}
		\int_{\Dint \setminus R} |f|^2 \d{x} \leq c e^{-2d\sqrt{\lambda}} \int_{R} |\Re(q)| |f|^2 \d{x}\leq c e^{-2d\sqrt{\lambda}} \int_{R} |q| |f|^2 \d{x}
\end{equation}	  
for all $f \in L^2(\Dint)$ solution of \; $\Delta f -  \lambda f = 0$.
\end{lemma}
\begin{proof}
See Lemma 2.3 in \cite{kirschitp} if $q$ is real and Lemma 3.14 in \cite{CCH} for $q$ complex.
\end{proof}
\begin{lemma}
\label{lem:isoA}
Assume that $q\in L^\infty(\Dint)$ is such that $q_{min}+1>0$ and either  $q_\star>0$ or $q^\star<0$. Then for  sufficiently large $\lambda > 0$, the operator $B_\lambda$ is an isomorphism form $\X$ onto itself.
\end{lemma} 
\begin{proof}
See Theorem 2.7 in \cite{kirschitp} if $q$ is real and Lemma 3.15 in \cite{CCH} for $q$ complex.
\end{proof}
\mfied{We now proceed with proving that  the operator $A_\lambda:X\to X$ is an isomorphism for sufficiently large $\lambda$. We adopt two different approaches for the cases $q^\star<0$ (Theorem \ref{lem:isoAlbda}) and $q_\star>0$ (Theorem \ref{neba})}.
\begin{theorem}
\label{lem:isoAlbda}
Assume that $q\in L^\infty(\Dint)$ is such that $q_{min}+1>0$ and $q^\star<0$. Then, for sufficiently large $\lambda > 0$, the operator $A_\lambda:X\to X$ is an isomorphism.
\end{theorem}
\begin{proof}
It is sufficient to prove the injectivity of $A_\lambda$ since $B_\lambda$ is an isomorphism from Lemma \ref{lem:isoA} and and $A_\lambda - B_\lambda$ is compact from Lemma \ref{lem:compAmA}.  To this end, assume that  $(v,f) \in \X$ is such that 
$A_{\lambda}(v, f) = 0$, i.e,  $a_{\lambda_j}(v,f;\phi,\psi) = 0$ for all $(\phi, \psi) \in \X$. As such,   $(v,f)$ satisfies 
\begin{equation}
\label{eqv}
\left\{
	\begin{array}{ccc}
		\Delta v - \lambda (q+1) v &=& q f+ p \opes{i\sqrt{\lambda}}(f) \\[1.5ex]
		\Delta f - \lambda f &=& 0 
	\end{array}  \qquad \qquad    \text{in} \qquad \Dint.
\right.
\end{equation}
  Multiplying the first equation of \eqref{eqv} with $\bar f$  we obtain
\begin{align*}
	\int_{\Dint} q|f|^2 +  \int_{\Dint} p \opes{i\sqrt{\lambda}}(f)  \ol{f} = \int_{\Dint} \Big( \Delta v  - \lambda (q+1) v\Big) \ol{f} 
	 =  \int_{\Dint} \Big(\Delta v - \lambda v\Big) \ol{f} - \int_{\Dint} q v \ol{f}.
\end{align*}
First observe that  $ \int_{\Dint} (\Delta v - \lambda v)  \ol{f} = 0$ since  $v \in H^2_0(\Dint)$ and $\Delta f -  \lambda f = 0$ in $\Dint$. Therefore
\begin{equation}
\label{eq:dk1}
	\int_{\Dint} q |f|^2 = -  \int_{\Dint} p \opes{i\sqrt{\lambda}}(f)  \ol{f} -  \lambda\int_{\Dint} q v \ol{f}.
\end{equation}
Multiplying now the first equation of \eqref{eqv} with $\bar v$ and integrating by parts we obtain
\begin{equation}
\label{eq:dk22}
	\int_{\Dint} | \nabla v |^2 + \lambda(q+1)|v|^2 \d x =  - \int_{\Dint} q f \ol{v} \d x -  \int_{\Dint} p \opes{i\sqrt{\lambda}} (f) \ol{v} \d x.
\end{equation}
Now assume that $q^\star<0$ and taking the real part of the above we write
\begin{equation}
\label{eq:dk2}
	\int_{\Dint} | \nabla v |^2 + \lambda(\Re(q)+1)|v|^2 =  \Re\left(- \int_{\Dint} q f \ol{v}  -  \int_{\Dint} p \opes{i\sqrt{\lambda}} (f) \ol{v}\right).
\end{equation}
 From \eqref{eq:dk1} and \eqref{eq:dk2} noting that $\lambda$ is real, we deduce that
\begin{eqnarray}
	\int_{\Dint} | \nabla v |^2 + \lambda (\Re(q) + 1) |v |^2  &+ & \frac{1}{\lambda}\int_{\Dint}  -\Re(q) |f|^2 \d{x} \label{contrac1}\\
	&=&  \Re \left(\int_{\Dint}p \opes{i\sqrt{\lambda}} (f)( \frac{1}{\lambda} \ol{f} -  \ol{v}) \d{x}\right).\nonumber
\end{eqnarray}
For large enough $\lambda>0$, we set $\rho: = \mbox{max}\left(\|q\|_{\infty}c e^{-2d\sqrt{\lambda}}, c e^{-2d\sqrt{\lambda}}\right) < 1$. Then Lemma \ref{lem:Kirsch} implies
\begin{equation}\label{rel:rho}
(1 - \rho) \int_{R} |\Re(q)| |f|^2 \d{x} \leq \int_{\Dint} -\Re(q) |f|^2 \d{x} \leq (1+\rho) \int_{R} |\Re(q)||f|^2 \d{x}.\end{equation}
The latter inequalities implies
\begin{eqnarray}
\int_{\Dint}  \lambda (\Re(q) + 1) |v |^2  &+&  \frac{1}{\lambda}\int_{\Dint}  - \Re(q) |f|^2 \d{x}  \label{eq:sgch1}\\
&\geq& \lambda  \int_{\Dint}   (\Re(q) + 1) |v |^2  + \frac{1 - \rho}{\lambda}\int_{R}  |\Re(q)| |f|^2 \d{x}.\nonumber  
\end{eqnarray}
From Lemma \ref{lem:pro1} we infer the existence of $c > 0$ such that $\|p \opes{i\sqrt{\lambda}} (f)\|_{L^2(\Dint)} \leq c e^{-\theta \sqrt{\lambda}} \|f\|_{L^2(\Dint)}$. We then have 
\begin{equation*}
	 \Big|\int_{\Dint} \Re\big[p \opes{i\sqrt{\lambda}} (f)(\frac{1}{\lambda} \ol{f} - \ol{v} )\big]  \d{x} \Big| 
	 \leq c e^{- \theta \sqrt{\lambda}} \| f\|_{L^2(\Dint)} \Big(\| v\|_{L^2(\Dint)} +\frac{1}{\lambda} \| f\|_{L^2(\Dint)} \Big)
\end{equation*}
Using Young's inequality then splitting the integral of $|f|^2$ into the domains $R$ and $\Dint \setminus R$ we finally have 
\begin{multline}
\label{eq:sgch2}
\left|\int_{\Dint} \Re\left(p \opes{i\sqrt{\lambda}} (f)(\frac{1}{\lambda} \ol{f} - \ol{v} )\right) \d{x} \right| \leq \frac{2c}{\lambda} e^{- \theta \sqrt{\lambda}} \|f\|^2_{L^2(\Dint)} + \frac{c\lambda}{4} e^{- \theta \sqrt{\lambda}} \|v\|^2_{L^2(\Dint)} \\
\leq \frac{2c}{\lambda} e^{- \theta \sqrt{\lambda}} (1 + \rho) \|f\|^2_{L^2(R)} + \frac{c\lambda}{4} e^{- \theta \sqrt{\lambda}} \|v\|^2_{L^2(\Dint)} \\
\leq  \frac{2c (1 + \rho)}{\lambda q_\star} e^{- \theta \sqrt{\lambda}} \| \sqrt{|\Re(q)|}f\|^2_{L^2(R)} + \frac{c\lambda}{4 (q_{min}+1)} e^{- \theta \sqrt{\lambda}} \| \sqrt{\Re(q)+1}v\|^2_{L^2(\Dint)}.
\end{multline}
For $\lambda$ sufficiently large such that
\[
	 \frac{2c(1+\rho)}{\lambda q_\star} e^{- \theta \sqrt{\lambda}} \leq \frac{1 - \rho}{\lambda} - \epsilon \qquad \text{and} \qquad \frac{c \lambda}{4 (q_{min}+1)} e^{- \theta \sqrt{\lambda}} \leq \lambda - \epsilon, 
\]
with $\epsilon >0$  {small enough} we finally obtain from  \eqref{contrac1}, \eqref{eq:sgch1} and \eqref{eq:sgch2} 
$$
\epsilon \| \sqrt{\Re(q)+1}v\|^2_{L^2(\Dint)} + \epsilon  \| \sqrt{|\Re(q)|}f\|^2_{L^2(R)} \le 0.
$$
This prove that $v = 0$ in $\Dint$ and $f=0$ in $R$. Inequality \eqref{rel:rho} then implies  $f=0$ in $\Dint$ which ends the proof. 
\end{proof}
\begin{theorem}\label{neba}
Assume that $q\in L^\infty(\Dint)$ is such that $q_{min}+1>0$ and  $q_\star>\sup_R(\Im(q))^2\geq 0$. Then, for sufficiently large $\lambda > 0$, the operator $A_\lambda:X\to X$ is an isomorphism.
\end{theorem}
\begin{proof}
Again similarly to the proof of Theorem \ref{lem:isoAlbda} it is sufficient to prove that $A_\lambda$ is injection, which will do by contradiction. Assume to the contrary that there exists a sequence $\lambda_j \to \infty$ and $(v_j, f_j) \in \X$ with $\|(v_j, f_j)\|_{L^2(\Dint)} = 1$ and $A_{\lambda_j} (v_j, f_j)=0$. Then $v_j \in H^2_0(\Dint)$ and $f_j \in L^2(\Dint)$ satisfy
\begin{equation}
\label{alambda1}
\Delta v_j - \lambda_j (q+1) v_j = q f_j+ p \opes{i\sqrt{\lambda_j}}(f_j) \qquad \text{and} \qquad \Delta f_j - \lambda_j f_j = 0.
\end{equation} 
For given $\epsilon >0$ small enough, from Lemma \ref{lem:pro1} and Lemma \ref{lem:Kirsch} and the fact that  $\|(v_j, f_j)\|_{L^2(\Dint)} = 1$, we have  that for  $\lambda_j >0$ large enough
\begin{eqnarray}
\| p \opes{i\sqrt{\lambda_j}}(f_j) \|_{L^2(\Dint)} \| f_j \|_{L^2(\Dint)} &+& \| \sqrt{|q|} f_j \|_{L^2(\Dint\setminus R)}^2 \label{wrong}\\
&+& \lambda_j  \| \sqrt{|q|} f_j \|_{L^2(\Dint\setminus R)}^2  \| \sqrt{|q|} v_j \|_{L^2(\Dint\setminus R)}^2 \le \epsilon^2.\nonumber
\end{eqnarray}
From  \eqref{eq:dk1}, we observe that 
\begin{eqnarray*}
\int_{R} \Re(q) | f_j|^2 \d x&=&\Re\left(-\lambda_j\int_{R}qv_j\overline{f_j} \d x\right)-\int_{\Dint \setminus R} \Re(q) | f_j|^2 \d x\\
&-&\Re\left( \int_{\Dint} p \opes{i\sqrt{\lambda}} (f_j) \ol{v_j} \d x+\lambda_j\int_{\Dint \setminus R}qv_j\overline{f_j} \d x\right)
\end{eqnarray*}
Estimating the last three terms from (\ref{wrong}) and using Cauchy-Schwarz inequality  we now obtain
$$
\int_{R}\Re(q) | f_j|^2d x \le \epsilon^2 + \lambda_j \left(\int_{R} \Re(q) | f_j|^2 \d x\right)^{1/2}  \left(\int_{R} \frac{|q|^2}{\Re(q)} | v_j|^2\d x\right)^{1/2}.
$$
This implies that
\begin{equation}
\label{NPT:1}
\left(\int_{R} \Re(q) | f_j|^2 \d x\right)^{1/2} \le \epsilon + \lambda_j  \left(\int_{R} \frac{|q|^2}{\Re(q)} | v_j|^2 \d x\right)^{1/2},
\end{equation}
(where we used for $A>0$ and $B>0$, $A^2\leq \epsilon^2+AB$ implies $(A-B/2)^2\leq \epsilon^2+B^2/4\leq (\epsilon +B/2)^2$  and taking the square root yields $A\leq \epsilon+B$). From  \eqref{eq:dk22}, we see that (using again $\|(v_j, f_j)\|_{\X} = 1$)
\begin{eqnarray}
\lambda_j \int_{R} \big(1+\Re(q)\big) |v_j|^2 \d x &\leq&  \lambda_j \int_{\Dint} \big(1+\Re(q)\big) |v_j|^2 \d x  \label{NPT:2}\\
&&\hspace*{-2cm} \leq\left(\int_{R} \Re(q) | f_j|^2 \d x\right)^{1/2} \left(\int_{R} \frac{|q|^2}{\Re(q)} | v_j|^2 \d x\right)^{1/2}  + \epsilon^2 \nonumber
\end{eqnarray}
since $(1 + \Re(q))>0$ in $\Dint$. Combining the last two inequalities yield
\begin{equation}
\lambda_j \int_{R} \big(1+\Re(q)\big) |v_j|^2 \d x \le \lambda_j  \int_{R} \frac{|q|^2}{\Re(q)} |v_j|^2 \d x + \epsilon \left(\int_{R} \frac{|q|^2}{\Re(q)} | v_j|^2 \d x\right)^{1/2}  + \epsilon^2.
\end{equation}
The assumption $q_\star>\sup\limits_R(\Im(q))^2\geq 0$ in particular implies $\inf_R\frac{\Re(q)-(\Im(q))^2}{\Re(q)}=\delta_0>0$, hence we have 
$$
\lambda_j  \delta_0 \int_{R} |v_j|^2 \d x \le \epsilon \left(\int_{R} \frac{|q|^2}{\Re(q)} | v_j|^2 \d x\right)^{1/2}  + \epsilon^2.
$$
Choosing $\lambda_j$ such that $ 4\lambda_j  \delta_0 \ge \|\frac{|q|^2}{\Re(q)}\|_{L^\infty(R)}$ we obtain that 
$$
\lambda_j \delta_0  \int_{R}  |v_j|^2 \d x \le 2 \epsilon \left(\lambda_j \delta_0 \int_{R} | v_j|^2 \d x\right)^{1/2}  + \epsilon^2.
$$
which implies 
$$
\left(\lambda_j \delta_0 \int_{R}  | v_j|^2 \d x\right)^{1/2}    \le 3 \epsilon
$$
and proves that $\dsp \lambda_j \delta_0 \int_{R}  | v_j|^2 \d x \to 0$ as $j \to \infty$. Coming back to
\eqref{NPT:2} we deduce that 
$$
\int_{\Dint} (1+\Re(q)) |v_j|^2 \d x \to 0 \mbox{ as } j \to \infty
$$
and from 
\eqref{NPT:1}, since $\Re(q)$ is positive definite in $R$, we deduce that $ \dsp   \int_{R}  | f_j|^2 \d x \to 0$ as $j \to \infty$. Lemma \ref{lem:Kirsch} then implies 
$$
\int_{\Dint} |f_j|^2 \d x \to 0 \mbox{ as } j \to \infty.
$$
The last two zero limits contradict the fact that $\|(v_j, f_j)\|_{L^2(\Dint)} = 1$, which proves injectivity.
\end{proof}

\noindent
Recall that solving the new interior transmission problem formulated in Definition \ref{nitp} is equivalent to solving
\begin{equation}\label{var11}
A_k(v,f)=\ell \qquad \mbox{in}\; \X
\end{equation}
(where we go back to the original eigen-parameter $k\in {\mathbb C}$ which is such that $\lambda=-k^2$). Let us fix $\lambda_0>0$ such that Theorem \ref{lem:isoAlbda} or Theorem \ref{neba} holds, i.e. $A_{\lambda_0}$ is an isomorphism and let $k_0=\sqrt{\lambda_0}i$. Then we can rewrite (\ref{var11}) as 
\begin{equation}
	(v,f) + A^{-1}_{k_0} (A_k - A_{k_0})(v,f) =A^{-1}_{k_0}\ell  \qquad \mbox{or} \qquad  (I-C_{k})(v,f)=A^{-1}_{k}\ell. 
\end{equation}
where $C_k:=A^{-1}_{k_0}(A_{k} - A_{k_0})$ is compact from Lemma \ref{lem:compAmA}. Thus the Fredholm alternative applies to (\ref{var11}). In particular a solution of the new interior transmission problem exists provided $k$ is not a new transmission eigenvalue defined in Definition \ref{ntp-def}. To show that the set of new transmission eigenvalues  is discrete we appeal to Fredholm Analytic Theory (see e.g. \cite{CK3}) since the mapping $k\mapsto C_k$ is analytic in ${\mathbb C}$.
Thus we have proven the following theorem:
\begin{theorem}\label{mainth}
Assume that the bounded function $n$ satisfies $\inf_{\Dint}{\Re(n)}>0$, $\Im(n)\geq 0$, and there exists a neighborhood $R$ of $\partial \Dint$ inside $\Dint$  such that either $\sup_{R}{\Re(n-1)}<0$ or  $\inf_{R}\Re(n-1)>\sup_R(\Im(n))^2\geq 0$. Then the new interior transmission formulated in Definition \ref{nitp} has a unique solution depending continuously on the data $\varphi$ and $\psi$ provided $k\in{\mathbb C}$ is not a new transmission eigenvalue defined  in Definition \ref{ntp-def}. In particular the set of new transmission eigenvalues in ${\mathbb C}$ is discrete (possibly empty) with $+\infty$  as the only possible accumulation point.
\end{theorem}
\noindent
Note that Theorem \ref{mainth} provides sufficient conditions under which Assumption \ref{as-nitp} hold. It is highly desirable to show if and when real new transmission eigenvalues exist, because for such real wave numbers our imaging algorithm introduced in the next section fails.
\section{A Differential Imaging Algorithm} \label{ch3subsec4}
We now apply all the results of Section 2 and Section 4 above to design an algorithm that provides us with the support of the defect $\domp$ without reconstructing $D_p$ or computing the Green's function of the periodic media.   \mfied{We follow the idea proposed in  \cite{Thi-Phong3} and build a differential imaging functional by comparing the application of the Generalized Linear Sampling algorithm to respectively the operators $N^\pm$ and  $N_q^{\pm}$}. The new results obtained in Theorem \ref{mainth} allow us to justify this algorithm for general location of $\omega$ (possibly multi-component).
\subsection{\mfied{Description and analysis of the algorithm} }
Throughout this section we assume that  Assumptions \ref{Ass:nk}, \ref{HypoLSM} and \ref{as-nitp}  hold. For sake of simplicity of presentation  we
only state the results when  the measurements operator $\open^+$ is available. Exactly the same
holds for the operator  $\open^-$ by changing everywhere the exponent  $+$  to $-$. For given $\phi$ and $a$ in $\ell^2(\Zd)$ we define the functionals
\newcommand{\pmp}{+}
\begin{equation}
	\begin{array}{ll}
		J^{\pmp}_{\alpha}(\phi,\aher) := \alpha(\open_{\sharp}^\pmp\aher,\aher) + \|\open^\pmp\aher - \phi\|^2, \\[1.5ex]
		J^{\pmp}_{\alpha,q}(\phi,\aher) :=
                \alpha(\open_{q,\sharp}^\pmp\aher,\aher) + \|\open_q^\pmp \aher - \phi\|^2
	\end{array}
\end{equation}
\mfied{with $\open_{q,\sharp}^\pmp:= I_q^* \open_{\sharp}^\pmp  I_q$.}
Let $\aher^{\alpha,z}$, $\aher_q^{\alpha,z}$ and $\tilde \aher^{\alpha,z}_q$ 
in $\ell(\Zd)$ verify (i.e. are minimizing sequences)
\begin{equation}
	\begin{array}{ll}
\displaystyle 		J^{\pmp}_{\alpha} (\sepgre{z},\aher^{\alpha,z}) \leq \inf_{\aher \in \ell^2(\Zd)} J^{\pmp}_{\alpha}(\sepgre{z},a) + c(\alpha)
		\\[12pt]
\displaystyle 		J^{\pmp}_{\alpha} (\sepgreq{z},\aher_q^{\alpha,z}) \leq \inf_{\aher \in \ell^2(\Zd)} J^{\pmp}_{\alpha}(\sepgreq{z},a) + c(\alpha)
		\\[12pt]
\displaystyle 			J^{\pmp}_{\alpha,q} (\opei^{*}_q \sepgreq{z},\tilde \aher_q^{\alpha,z}) \leq \inf_{\aher \in \ell^2(\Zd)} J^{\pmp}_{\alpha,q}(\opei^{*}_q \sepgreq{z},a) + c(\alpha)
	\end{array}
\end{equation}
with $\frac{c(\alpha)}{\alpha} \to 0$ as $\alpha \to 0$.  Here $\sepmgre{z}$  are  the Rayleigh coefficients of $\Phi(1,z)$ (i.e. $\Phi(n_p;z)$ defined by (\ref{phi}) with with $n_p=1$) given by (\ref{lunch}) and $\sepmgreq{z}$  are the Rayleigh coefficients of $\greq{z}$ given by (\ref{hhh}).

\noindent
The standard analysis of the generalized linear sampling method (see e.g. \cite[Section 2.2]{CCH}) making use of the factorization of  $\open_{\sharp}$  in  Theorem \ref{TheoFactorization} along with  all the properties of the involved operators developed in Section 2.2 and Section 3 imply the following results (see also  \cite{Thi-Phong3} and  \cite{tpnguyen} for detailed proofs)
\begin{lemma}\label{glsm}
\begin{enumerate}
\item  $\zg \in D$ if and only if $ \lim \limits_{\alpha \to 0} (\open_{\sharp}^+\aher^{\alpha,z},\aher^{\alpha,z})< \infty$. Moreover, if $\zg \in D$ then $\opeh^+ \aher^{\alpha,z} \to v_z$ in $L^2(D)$ where $(u_z, v_z)$ is the solution of problem \eqref{oitp} with $\ftd= \Phi(1;z)$
and $\ftn = \partial \Phi(1;z)/\partial \nu$ on $\partial D$.
\item  $\zg \in D_p$ if and only if $ \lim \limits_{\alpha \to 0} (\open_{\sharp}^+\aher_q^{\alpha,z},\aher_q^{\alpha,z})< \infty$. Moreover, if $\zg \in D_p$ then $\opeh^+ \aher_q^{\alpha,z} \to v_z$ in $L^2(D)$ where $(u_z, v_z)$ is the solution of problem \eqref{oitp} with $\ftd= \greq{z}$
and $\ftn = \partial \greq{z}/\partial \nu$ on $\partial D$.
\item $\zg \in \Dtot_p$ if and only if $ \lim \limits_{\alpha \to 0} (\open_{q,\sharp}^+\tilde \aher^{\alpha,z}_q,\tilde \aher^{\alpha,z}_q) < \infty$. Moreover, if $\zg \in \Dtot_p$ then $\opeh_q^+ \tilde \aher^{\alpha,z}_q \to h_z$ in $L^2(D)$ where $h_z$ is defined by 
\begin{equation}
\label{def:fz}
\left.
\begin{array}{ll}
h_z = \left\{ \begin{array}{cl}
	- \greq{z} \quad  &\text{in}  \quad   \Dint_p \\[1.25ex]
	v_z  \quad   &\text{in} \quad   \Dpoutp
\end{array} \right. \qquad  \text{if} \quad   z \in \Dpoutp \\[4.ex]
 h_z = \left\{ \begin{array}{cl}
	\widehat{v}_z \quad  & \text{in} \quad  \Dint_p \\[1.25ex]
	- \greq{z} \quad  &\text{in}  \quad   \Dpoutp
\end{array} \right. \qquad  \text{if} \quad  z \in \Dint_p
\end{array}
\right.
\end{equation}
 where $(u_z, v_z)$ is the solution of problem \eqref{oitp} with $\ftd= \greq{z}$
and $\ftn = \partial \greq{z}/\partial \nu$ on $\partial D$ and $(\widehat{u}_z,\widehat{v}_z)$ is $\alpha_q$-quasi-periodic extension of  the solution $(u,f)$ of the new interior transmission problem  in Definition \eqref{nitp} with $\ftd= \greq{z}$
and $\ftn = \partial \greq{z}/\partial \nu$ on $\partial \Dint$.
\end{enumerate}
\end{lemma}

\begin{proof}
	The proof of the items $(i)$ and $(ii)$, we refer to \cite{Thi-Phong3}. The proof of items $(iii)$ is a direct application of Theorem A.4 in \cite{Thi-Phong3} in combination with Theorem \ref{Lem:injectiveG_q2}. 
	
\end{proof}

\noindent
We then consider the following imaging functional that characterizes $\Dint$,
\begin{equation} \label{IndFuncDiffIm}
\I^{\pmp}_\alpha(z) = \left({(\open_{\sharp}^\pmp a^{\alpha,z}, a^{\alpha,z}) \left(1 + \frac{(\open_{\sharp}^\pmp a^{\alpha,z}, a^{\alpha,z})}{D^{\pmp}(a^{\alpha,z}_q, \tilde a^{\alpha,z}_q)}\right)}\right)^{-1}
\end{equation}
where for $a$ and $b$ in $\ell^2(\Zd)$,  
$$ D^{\pmp}(a, b):= \left(\open_{\sharp}^\pmp (a - \opei_q b), (a - \opei_q b)\right). $$ 
Before giving the main theorem for the characterization of the defect we need the following assumption.
\begin{asp}
\label{Ass:def-itp}
The refractive indexes $n$, $n_p$ and wave-number $k$ are such that 
\begin{equation} \label{defitp}
\left\{ \begin{array}{lll}
\Delta {\eg} + k^2 n {\eg} = 0 \quad& \mbox{ in } \; \domp, 
\\[6pt]
\Delta \vgg + k^2 n_p \vgg = 0  \quad &\mbox{ in } \; \domp,
\\[6pt]
 {\eg} - \vgg= 0 \quad &\mbox{ on } \; \partial \domp,
\\[6pt]
\partial ({\eg} - \vgg)/\partial \nu = 0 \quad &\mbox{ on } \; \partial \domp
\end{array}\right.
\end{equation}
has only the trivial solution. 
\end{asp} 
This assumption  is satisfied if $n - n_p$ does not change sign in a neighborhood of the boundary of domain $\domp$, or if either $n$ or $n_p$ have non-zero imaginary part (see  e.g \cite{CCH}).

\begin{theorem} \label{DiffImagFunc}
Under Assumptions \ref{Ass:nk}, \ref{HypoLSM},\ref{as-nitp} and \ref{Ass:def-itp}, we have that
$$z \in D\setminus \Dpoutp \mbox{\;\;  if and only if \;\; } \lim_{ \alpha \to 0} \I^{\pmp}_\alpha(z) > 0.$$ 
(Note that $D\setminus \Dpoutp = \domp \cup \Dpint_p$ contains the physical defect $\omega$ and \mfied{$\Dpint_p := D_p \setminus \Dpoutp$} the components of $D_p$ which have nonempty intersection with the defect).
\end{theorem}
\begin{proof}
{\it Case 1:  $z\notin  D\setminus \Dpoutp$}. If $z \notin D$ then from Lemma \ref{glsm}$(i)$ we have that
$(\open_{\sharp}^\pmp \aher^{\alpha,\zg} ,\aher^{\alpha,\zg}) \to + \infty $ as $\alpha \to 0$ and therefore 
$\displaystyle \lim_{\alpha \to 0} {\I^+_\alpha}(\zg) = 0.$
\medskip \\
If $z \in \Dpoutp$, let  $(u_z, v_z) \in L^2(D) \times L^2(D)$ be the unique solution of (\ref{oitp}) with $\varphi:= \Phi_q(\cdot;z))|_{\partial D}$ and $\psi=\partial   \Phi_q(\cdot;z))/\partial \nu|_{\partial D}$  then from Lemma \ref{glsm}$(ii)$, $\opeh^+ \aher^{\alpha,z} \to v_z$. Note that in this case $u_z = 0$ and $v_z:= -  \Phi_q(\cdot;z)$ outside $\Dpoutp$. 

\noindent Now let  $h_z \in L^2(\Dtot_p)$ defined by \eqref{def:fz}. Then from Lemma \ref{glsm}$(iii)$, $\opeh^\pmp_q\tilde a^{\alpha,z}_q \to h_z$ in $L^2(D)$.
 Furthermore, by definition of $h_z$ and the fact that $v_z = - \Phi_q(\cdot,z)$ in $\Dint_p$, we have that $v_z$ coincide with $h_z$ in $D$.  From the factorization of $\open_\sharp^\pmp$ and the definition of $\opeh^\pmp_q$ we  have
\begin{align*}
D^{\pmp}(a^{\alpha,z}_q , \tilde a^{\alpha,z}_q) &=
\left(\open^+_{\sharp} (a^{\alpha,z}_q -  I_q \tilde a^{\alpha,z}_q),  a^{\alpha,z}_q - I_q \tilde a^{\alpha,z}_q\right) \\
 &=\left(\opet_{\sharp} (\opeh^\pmp  a^{\alpha,z}_q -  \opeh^\pmp_q \tilde a^{\alpha,z}_q), \opeh^\pmp  a^{\alpha,z}_q -  \opeh^\pmp_q \tilde a^{\alpha,z}_q\right)  \\
 & \le
        \|\opet_{\sharp}\| \| \opeh^\pmp  a^{\alpha,z}_q -  \opeh^\pmp_q \tilde
        a^{\alpha,z}_q\|_{L^2(D)}^2 \to 0, \quad \mbox{ as } \alpha \to 0. 
\end{align*}
Since from Lemma \ref{glsm}$(i)$, $(\open^+_{\sharp} a^{\alpha,z},
a^{\alpha,z}) $ remains bounded as $\alpha \to 0$, we can finally  conclude  from the above that  $$\displaystyle \lim_{\alpha \to 0}
{\I^{\pmp}_\alpha}(\zg) = 0  \; \mbox{ if } z \in \Dpoutp.$$

\medskip

\noindent
{\it Case 2:  $z \in D \setminus \Dpoutp \subset \Dint_p$}. Then again  by Lemma \ref{glsm}$(i)$, $(\open^+_{\sharp} a^{\alpha,z},
a^{\alpha,z}) $ remains bounded and by Lemma \ref{glsm}$(iii)$, $(\open_{q,\sharp} \tilde a^{\alpha,z}_q,
\tilde a^{\alpha,z}_q) $ remain bounded. Using the factorization  of $\open^+_{q,\sharp}$ and the fact that  $\open_{q,\sharp}^\pmp = I^*_q \open_\sharp^\pmp I_q$ we can write 
 \[
 (\open^{\pmp}_{\sharp} \opei_q \tilde a^{\alpha,z}_q, \opei_q \tilde a^{\alpha,z}_q ) \to (T_\sharp h_z, h_z) < +\infty,
 \]
 where  $h_z \in L^2(\Dtot_p)$ defined in \eqref{def:fz}. In this case, we need to consider the location of $z$ in two sub-domain.
 
 \noindent
 {\it First}, if $z \in \domp \setminus D_p$ (the part of defect outside periodic domain), then by Lemma \ref{glsm}$(ii)$ $(\open^{\pmp}_{\sharp}
 a^{\alpha,z}_q, a^{\alpha,z}_q ) \to +\infty$ \; as \; $\alpha \to 0$. This implies,
  \[
D^{\pmp}(a^{\alpha,z}_q , \tilde a^{\alpha,z}_q) \geq \left|(\open^+_{\sharp} a^{\alpha,z}_q,
a^{\alpha,z}_q)-(\open^{\pmp}_{\sharp} \opei_q \tilde a^{\alpha,z}_q, \opei_q \tilde a^{\alpha,z}_q )\right|  \to +\infty.
\] We then conclude that 
$$ \lim_{\alpha \to 0} {\I^{\pmp}_\alpha}(\zg)  \neq 0  \quad \mbox{ if } \;  z \in \domp \setminus D_p.$$
  
\noindent  
{\it Next}, if $z \in D_p \cap \Dint_p = \Dpint_p$, again by Lemma \ref{glsm}$(ii)$ $(\open^{\pmp}_{\sharp}
 a^{\alpha,z}_q, a^{\alpha,z}_q ) \to (T_\sharp  v_z,  v_z)$ where $v_z \in L^2(D)$ is  defined in Lemma \ref{glsm}$(ii)$. In this case $v_z  = - \Phi_q(\cdot, z)$ outside $D \setminus \Dpoutp$, which implies that $h_z = v_z$ in $\Dpoutp$. On the other hand, $h_z|_{\Dpint_p}$ is $\alpha_q-$quasi-periodic \mfied{with period $L$} while  $v_z|_{\Dpint_p}$ is not $\alpha_q-$quasi-periodic \mfied{with period $L$} (recall that $(u_z,v_z)$ be solution of \eqref{oitp} defined in $\Dpoutp$ with  $\ftd= \Phi_q(1;z)$
and $\ftn = \partial \Phi_q(1;z)/\partial \nu$ on $\partial (D \setminus \Dpoutp$).  \mfied{ Indeed, assume to the contrary that $v_z|_{\Dpint_p}$ is $\alpha_q-$quasi periodic with period $L$. For a fixed arbitrary  $0 \neq m \in \Z_M$, let us define 
 \[
 	\widetilde{u}_z(x):= e^{-i\alpha_qmL} u_z(x + mL), \quad \text{for} \quad x\in \Dpint,
 \] hence  $(\widetilde{u}_z, v_z)$ satisfy
  \begin{equation} \label{eq:defpf1}
\left\{ \begin{array}{lll}
\Delta {\widetilde{u}_z} + k^2 n_p\widetilde{u}_z = 0 \quad& \mbox{ in } \; \Dpint, 
\\[6pt]
\Delta v_z+ k^2  v_z = 0  \quad &\mbox{ in } \; \Dpint,
\\[6pt]
 \widetilde{u}_z - v_z= \Phi_q \quad &\mbox{ on } \; \partial \Dpint,
\\[6pt]
\partial (\widetilde{u}_z - v_z)/\partial \nu = \partial/\partial \nu \Phi_q  \quad &\mbox{ on } \; \partial \Dpint.
\end{array}\right.
\end{equation}
Next, we let
\[
	\widehat{n} = \Big\{ \begin{array}{cl}
 							n_p &\quad \text{in} \quad \Dpint  \\
 							1 &\quad \text{in} \quad \domp \setminus \Dpint
 	\end{array}  \quad \text{and} \quad \widehat{u}_z  = \Big\{ \begin{array}{cl}
 							\widetilde{u}_z &\quad \text{in} \quad \Dpint \\
 							\Phi_q + v_z &\quad \text{in} \quad \domp \setminus \Dpint
 	\end{array}
\] and observe that $(\widehat{u}_z, v_z)$ verify
 \begin{equation} \label{eq:defpf2}
\left\{ \begin{array}{lll}
\Delta \widehat{u}_z + k^2 \widehat{n} \widehat{u}_z = 0 \quad& \mbox{ in } \; \Dint,
\\[6pt]
\Delta v_z+ k^2 v_z = 0  \quad &\mbox{ in } \; \Dint,
\\[6pt]
 \widehat{u}_z - v_z = \Phi_q \quad &\mbox{ on } \; \partial \Dint
\\[6pt]
\partial (\widehat{u}_z - v_z)/\partial \nu = \partial  \Phi_q / \partial \nu \quad &\mbox{ on } \; \partial \Dint,
\end{array}\right.
\end{equation}
where we use the fact that from (\ref{eq:defpf1}) the Cauchy data of  ${\widetilde{u}_z}$ and $\Phi_q + v_z$ coincide on $\partial \Dpint\cap \domp$. From the definition of $u_z$ and  \eqref{eq:defpf2} we have that $(u_z,\widehat{u}_z)$ is a solution to 
  \begin{equation} \label{eq:defpf3}
\left\{ \begin{array}{lll}
\Delta u_z + k^2 n u_z = 0 \quad& \mbox{ in } \; \Dint,
\\[6pt]
\Delta \widehat{u}_z+ k^2 \widehat{n} \widehat{u}_z = 0  \quad &\mbox{ in } \; \Dint,
\\[6pt]
 u_z -\widehat{u}_z = 0 \quad &\mbox{ on } \; \partial \Dint
\\[6pt]
\partial ({u}_z - \widehat{u}_z)/\partial \nu = 0  \quad &\mbox{ on } \; \partial \Dint.
\end{array}\right.
\end{equation}
Since $n=n_p=\widehat{n}$ in  $\Dint \setminus \domp$, then $u_z$ and $\widehat{u}_z$  satisfy the same equation in $\Dint \setminus \domp$  and share the same Cauchy data on $\partial \Dint \setminus \ol{\domp}$, hence by the unique continuation $u_z = \widehat{u}_z$ in $\Dint \setminus \domp$. Therefore, $(u_z ,\widehat{u}_z)$ is a solution of (\ref{defitp}) and Assumption \ref{Ass:def-itp} implies that $u_z = \widehat{u}_z = 0$ in $\domp$, and  consequently by definition $v_z+\Phi_q=0$ in $\domp\setminus \Dpint$ (recall $\Dint = \Dpint \cup \domp$). On the other hand, we have that  $v: =v_z+\Phi_q$ satisfies $\Delta v + k^2 v = 0$ in $\Dint \setminus \{z\}$, hence unique continuation implies that $v_z = -\Phi_q$ in $\Lambda \setminus \{z\}$. But this is a contradiction, since by the interior regularity of solutions of Helmholtz equation  $v_z$ is infinitely many times differentiable whereas $\Phi$ has a singularity at $z$.} This proves  that $v_z|_{\Dpint_p}$ is $\alpha_q-$quasi periodic with period $L$, and hence  $v_z \neq h_z$ in $D \setminus \Dpoutp$. We now have from the estimate
$$
\mfied{C \| \opeh^\pmp  a^{\alpha,z}_q -  \opeh^\pmp_q \tilde
        a^{\alpha,z}_q\|_{L^2(D)}^2 }\leq D^{\pmp}(a^{\alpha,z}_q , \tilde a^{\alpha,z}_q) \leq \|\opet_{\sharp}\| \| \opeh^\pmp  a^{\alpha,z}_q -  \opeh^\pmp_q \tilde
        a^{\alpha,z}_q\|_{L^2(D)}^2 
$$
where $C$ is the coercivity constant associated with $\opet_{\sharp}$, that $D^{\pmp}(a^{\alpha,z}_q , \tilde a^{\alpha,z}_q)$ is bounded and not goes to $0$  as $\alpha \to 0$. Thus
$$ \lim_{\alpha \to 0} {\I^{\pmp}_\alpha}(\zg)  \neq 0 \quad \mbox{ if } z \in \Dpint_p.$$
This ends the proof of the theorem. 
\end{proof}
\noindent
{Theorem \ref{DiffImagFunc} shows that the  functional  $\I^{\pmp}_{\alpha}(z)$  provides an indicator function for $D\setminus \Dpout$, i.e. the defect and the periodic components of the background that intersects $\omega$. However,  the proof of Theorem \ref{DiffImagFunc} indicates  that although the values of  $\I^{\pmp}_{\alpha}(z)$ are positive only in  $D\setminus \Dpout$,  they \mfied{are smaller} for $z \in \Dpint_p$  compared to the  values of $\I^{\pmp}_{\alpha}(z)$ for  $z \in \domp \setminus D_p$. Therefore, if the defect $\domp$ has nonempty intersection with the periodic background then the reconstruction of $\Dpint_p$ (i.e. the components that have no  nonempty intersection with the defect) are not displayed as clearly as the reconstruction of the  part of defect outside $D_p$. This is illustrated in the following numerical experiments.}

\noindent
We recall that exactly the same can be shown for down-to-up incident field, by simply replacing the upper index $+$ with $-$.  It is also possible to handle the case with noisy data, and we refer the reader to  \cite{Thi-Phong3} and  \cite{tpnguyen} for more detailed discussion.

\subsection{Numerical Experiments}
\label{sec:numerical}
We conclude by showing several numerical examples to test  our differential imaging algorithm. We limit ourselves to examples in ${\mathbb R}^2$. The data  is computed   with both  down-to-up and  up-to-down plane-waves  by solving the forward scattering problem based on the spectral discretization  scheme of  the volume integral formulation of the problem presented in \cite{Thi-Phong2}.

\noindent
Let us  denote by 
$$\Zd_{inc}:=\{j = q + M\ell, \;  q \in \Zd_M, \; \ell \in \Zd \; \text{and} \; \ell \in \interd{-N_{min}}{N_{max}} \}$$
the set of indices for the incident waves (which is also the set of indices for measured Rayleigh coefficients). The values of all  parameters used in our experiments  will be indicated below. The discrete version of the operators $\open^{\pm}$ are given by the $N_{inc} \times N_{inc}$ matrixes
\begin{equation}
\label{discre:MatrixN}
\open^{\pm} := \left(\raycoefpm{u^s}{\ell;j}\right)_{\ell,j \in \Zd_{inc}}.	
\end{equation}
Random noise is added to the data. More specifically, we in our computations we use 
\begin{equation}
\label{discre:MatrixNnoise}
\open^{\pm,\delta}(j,\ell): = \open^{\pm}(j,\ell)\big(1 + \delta A(j,\ell)\big), \quad \forall (j,\ell) \in \Zd_{inc}\times \Zd_{inc}
\end{equation}
where $A = (A(j,\ell))_{N_{inc}\times N_{inc}}$ is a matrix of uniform complex
random variables with real and imaginary parts in $[-1,1]^2$ and  $\delta > 0$
is the noise level. In our examples we take $\delta = 1\%$.

\noindent
For  noisy data, one needs to redefine the functionals $J^{\pmp}_{\alpha}$
and $J^{\pmp}_{\alpha, q}$ as \begin{equation}
	\begin{array}{ll}
		J^{\pmp,\delta}_{\alpha}(\phi,\aher) := \alpha\left(
                  (\open_{\sharp}^{\pmp,\delta}\aher,\aher) + \delta
                  \|\open_{\sharp}^{\pmp,\delta}\|
                  \|a\|^2\right) +  \|\open^{\pmp, \delta}\aher - \phi\|^2, \\[1.5ex]
		J^{\pmp,\delta}_{\alpha,q}(\phi,\aher) :=\alpha\left(
                  (\open_{\sharp}^{\pmp,\delta}\opei_q \aher,\opei_q\aher) + \delta
                  \|\open_{\sharp}^{\pmp,\delta}\|
                  \|a\|^2\right) +  \|\open^{\pmp, \delta}_q\aher - \phi\|^2
	\end{array}
\end{equation}
 We then consider  $\aher^{\alpha,z}_\delta$, $\aher_{q,\delta}^{\alpha,z}$ and $\tilde \aher^{\alpha,z}_{q,\delta}$ 
in $\ell(\Zd)$ as the minimizing sequence  of, respectively,  
$$
J^{\pmp,\delta}_{\alpha}(\sepgre{z},\aher), \;
J^{\pmp,\delta}_{\alpha}(\sepgreq{z},\aher) \mbox{ and } J^{\pmp,\delta}_{\alpha,q}(\sepgreq{z},\aher).
$$
The noisy indicator function takes the form
\begin{equation} \label{IndFuncDiffImNoisy}
\I^{\pmp,\delta}_\alpha(z) = \left({\mathcal{G}^{+,\delta}(a^{\alpha,z}_\delta)\left(1 +
    \frac{\mathcal{G}^{+,\delta}(a^{\alpha,z}_\delta)}{D^{\pmp, \delta}(a^{\alpha,z}_{q,\delta}, \tilde a^{\alpha,z}_{q,\delta})}\right)}\right)^{-1}
\end{equation}
where for $a$ and $b$ in $\ell^2(\Zd)$,  
$$ D^{\pmp,\delta}(a, b):= \left(\open_{\sharp}^{\pmp,\delta} (a - \opei_q b),
  (a - \opei_q b)\right) $$
and 
$$
\mathcal{G}^{+,\delta}(\aher) :=  (\open_{\sharp}^{\pmp,\delta}\aher,\aher) + \delta
                  \|\open_{\sharp}^{\pmp,\delta}\|
                  \|a\|^2.
$$

\noindent
Defining in a similar way the  indicator function
$\I^{-,\delta}(z)$ corresponding to up-to-down incident waves, we use the following indicator function in our numerical examples
$$
\I^\delta(z): = \I^{+,\delta}(z) + \I^{-,\delta}(z).
$$
In the following example, we consider the same periodic background $D_p$, in which each cell consists  of two circular components, namely the  discs with radii $r_1$, $r_2$ specified below. The physical parameters are set as 
\begin{equation}
\label{Num:parameters}
k = \pi/3.14, \; \; n_p = 2 \text{ inside the discs and } \; n_p = 1 \;  \text{otherwise}.  
\end{equation}
Letting $\lambda := 2\pi/k$ be the wavelength,  the geometrical parameters are set as
\begin{equation}
\label{Num:parameters1}
\text{the period}  \, L = \pi \lambda, \; \text{the width of the layer} \, h  =  1.5 \lambda, \; r_1 = 0.3 \lambda,  \; \text{and} \; r_2 = 0.4 \lambda.
\end{equation}
Finally we choose the truncated model 
\begin{equation}
\label{Num:parameters2}
		 M = 3, \;\;  N_{min} = 5\;  \mbox{ and } \; N_{max} = 5 \mfied{ \mbox{ and } q = 1}
\end{equation} 
The reconstructions are displayed by plotting  the indicator function $\I^\delta(z)$. 

\subsubsection*{Example 1:}
In the first example, we consider the perturbation $\domp$ be a disc of radius $r_\omega = 0.25 \lambda$ with the refractive index $n = 4$ and located in the component of radii $r_2$ (see  Figure \ref{Composite1_structure}-left).  In this case we reconstruct the whole component which contains the defect and its $L$-periodic copies (see Figure \ref{Composite1_structure}-right). We remark that this is the unfortunate case when it is not possible to determine in which period is the defect located. This is in accordance with what Theorem \ref{DiffImagFunc} predicts.
\begin{figure}[H] 
\centerline{ \begin{tabular}{cc}
\includegraphics[width=0.45\textwidth]{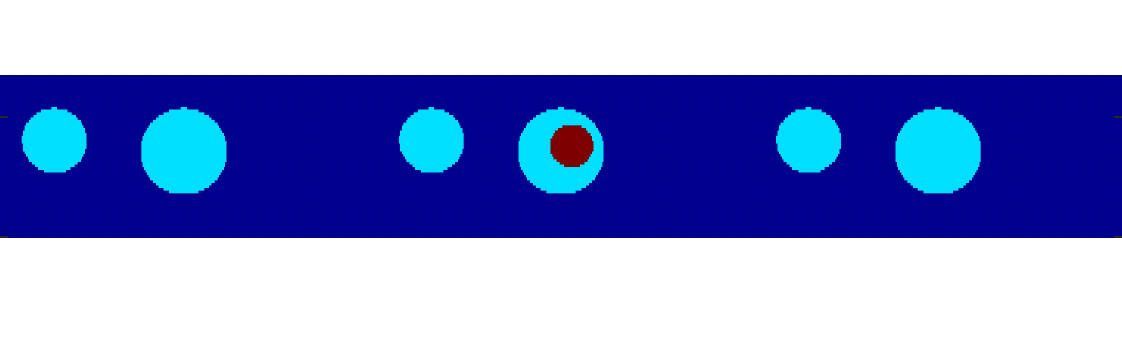}  & \includegraphics[width=0.45\textwidth]{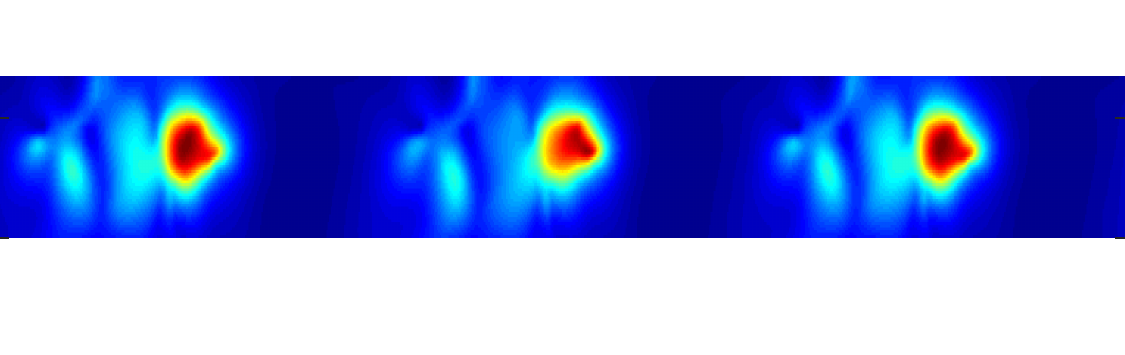} 
\end{tabular}}
\caption{Left: The exact geometry for Example 1. Right: The reconstruction using $z \mapsto
  \I^\delta(z)$}
\label{Composite1_structure}
\end{figure}
\subsubsection*{Example 2 (a).}
In the second example, we consider the perturbation $\domp$ as in  Example 1 but now located such that $\domp$ has nonempty intersection with $D_p$ but not included in $D_p$ (see  Figure \ref{Composite2.1_structure}-left). The reconstruction is represented in Figure \ref{Composite1_structure}-right. This example also illustrate that the value of $\I^{\delta}(z)$ much bigger when $z$ in $\domp \setminus D_p$ than $z$ in $\Dpint_p$. Now we are able to find in which period the defect is located and recover clearly the part of the defect outside component.
\begin{figure}[H] 
\centerline{\begin{tabular}{cc}
\includegraphics[width=0.45\textwidth]{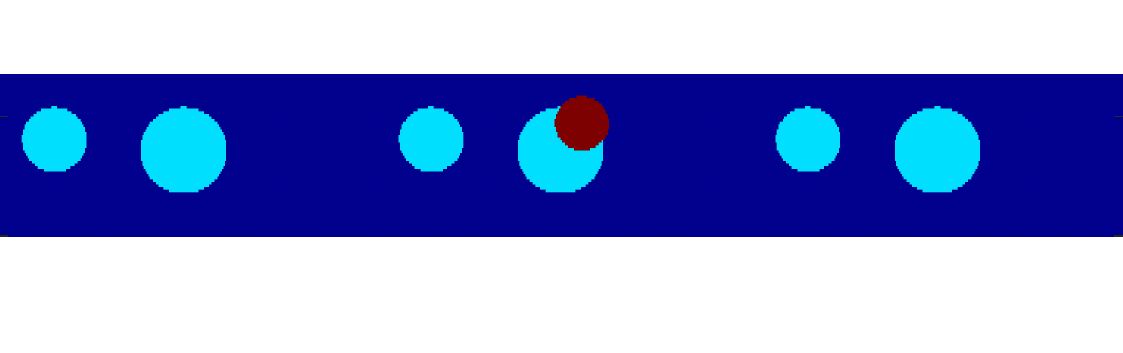} & \includegraphics[width=0.45\textwidth]{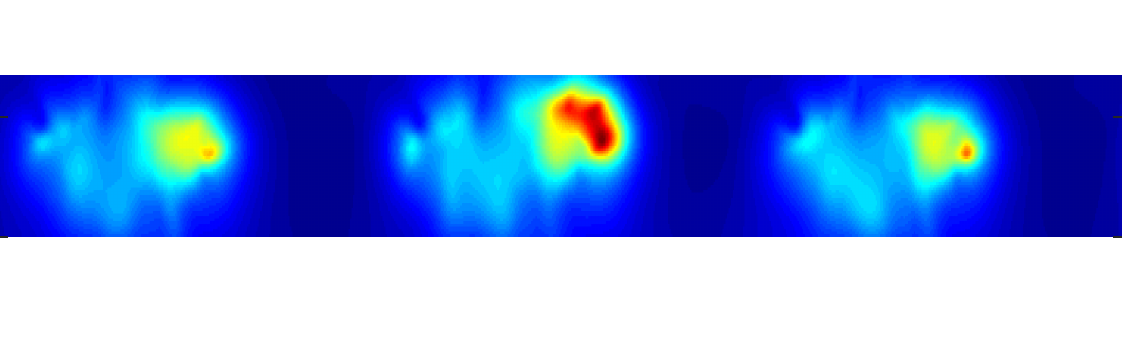}
\end{tabular}}
\caption{Left: The exact geometry for Example 2. Right: The reconstruction of the local perturbation using $z \mapsto
  \I^\delta(z)$}
\label{Composite2.1_structure}
\end{figure}

\subsubsection*{Example 2 (b).} In Figure \ref{Composite2.2_structure} we consider same configuration  as in Figure \ref{Composite2.1_structure} and change only the refractive index of the defect which now is inhomogeneous. In particular, the refractive  index of the defect is  $n = 4$ in $\domp \cap D_p$ and $n = 3$ in $\domp \setminus D_p$. The  reconstruction is represented in Figure \ref{Composite2.2_structure}--right

\begin{figure}[H] 
\centerline{\begin{tabular}{cc}
\includegraphics[width=0.45\textwidth]{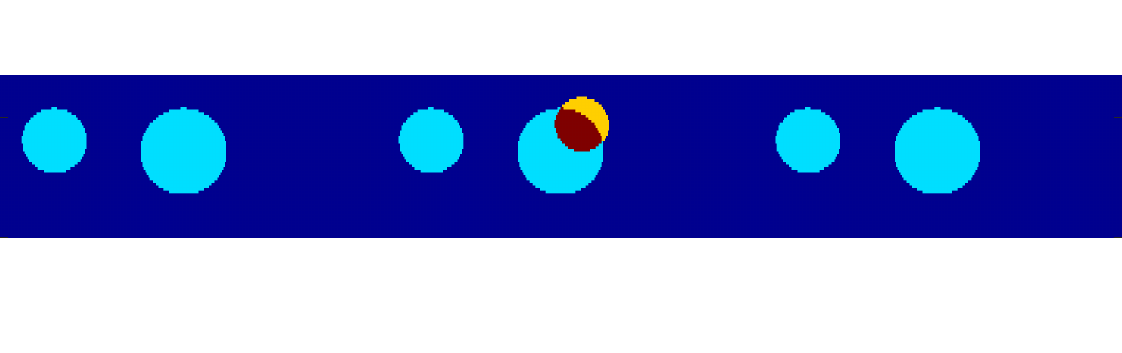} &
\includegraphics[width=0.45\textwidth]{Figures/0425_3PCompoPerturbationIntV1} \\
\end{tabular}}
\caption{Left: The exact geometry for Example 2. Right: The reconstruction of the local perturbation using $z \mapsto
  \I^\delta(z)$}
\label{Composite2.2_structure}
\end{figure}

\subsubsection*{Example 3.}
Example 3 just show that when the defect has no intersection with the periodic background, the indicator function $\I^{\delta}(z)$ allows to reconstruct the true defect including its true location. Here the defect is a disc of $r_\omega=0.2 \lambda$  with $n=3$.  More examples of this case can be found in  \cite{Thi-Phong3}.
\begin{figure}[H] 
\centerline{\begin{tabular}{cc}
\includegraphics[width=0.45\textwidth]{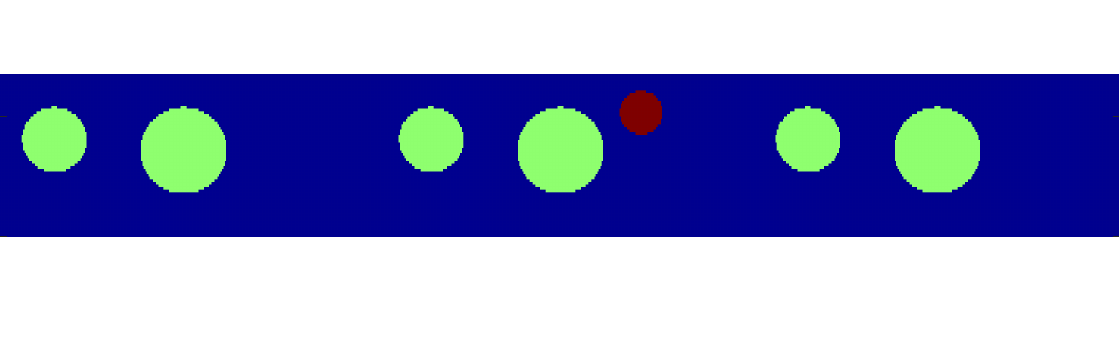} &
\includegraphics[width=0.45\textwidth]{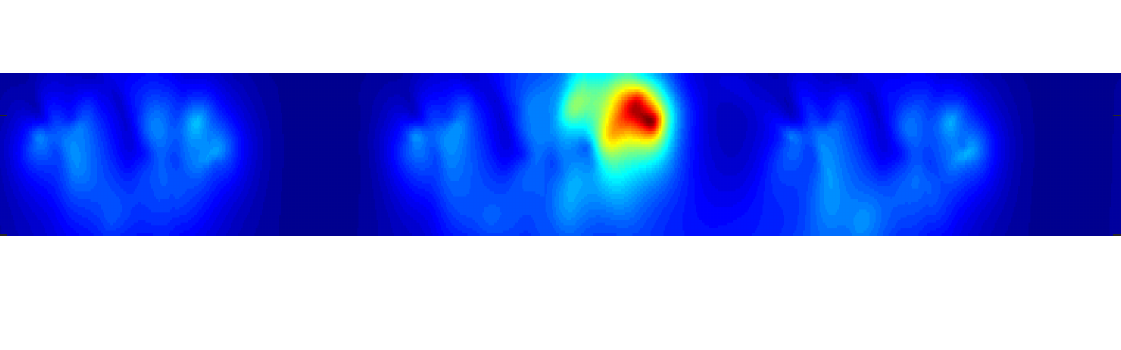}
\end{tabular}}
\caption{Left: The exact geometry for Example 3. Right: The reconstruction using $z \mapsto
  \I^\delta(z)$}
\label{Composite3_structure}
\end{figure}
\noindent
The last two examples, present the case where the defect consists of two disconnected components. 
\subsubsection*{Example 4.}
In this example the multicomponent defect has one component intersecting $D_p$ and one component outside $D_p$.  The true geometry is given in Figure \ref{IntMixed_structure} -left and the reconstruction in Figure \ref{IntMixed_structure}-right. Here $\omega_1$ is the disc intersecting $D_p$ and $\omega_2$ the other. The parameters are    $r_{\omega_1}=r_{\omega_2}=0.2 \lambda$.  $\omega_1$  is inhomogeneous with refractive index $n=4$in the part inside $D_p$ and $n=3$ in the part  inside $D_p$, wheres the refractive index of $\omega_2$ is $n=2.5$.
\begin{figure}[H] 
\centerline{\begin{tabular}{cc}
\includegraphics[width=0.45\textwidth]{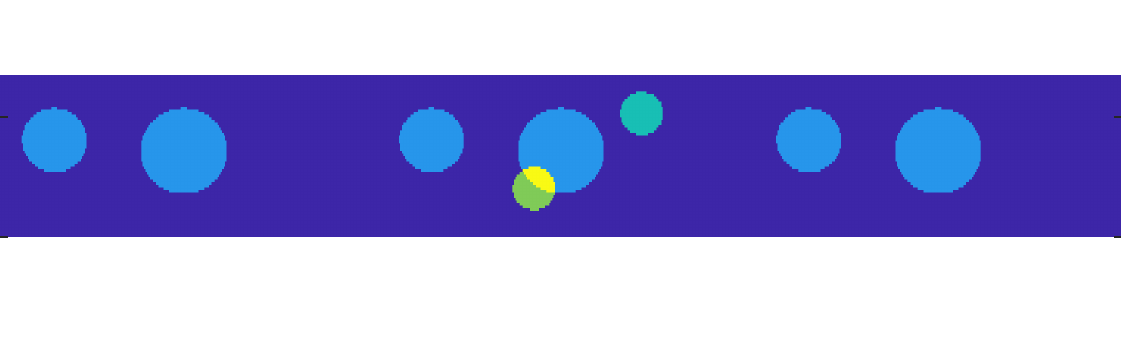} &
\includegraphics[width=0.45\textwidth]{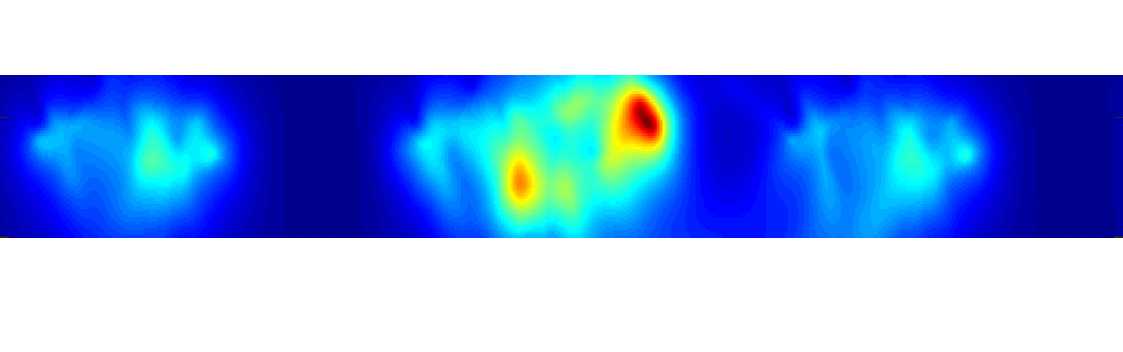}
\end{tabular}}
\caption{Left: The exact geometry for Example 4 Right: The reconstruction using $z \mapsto
  \I^\delta(z)$. }
\label{IntMixed_structure}
\end{figure}

\subsubsection*{Example 5.}
In this last example,  the defect has two disconnected component such that one  component included in one component of $D_p$ and the other lies outside $D_p$.  An illustration of exact geometry is given in Figure \ref{IncluMixed_structure} -left and the reconstruction in Figure \ref{IncluMixed_structure}-right. Here keeping the same notations for $\omega_1$ and $\omega_2$ as in Example 5,  we choose $n=4$ the refractive index of $\omega_1$ and $n=2.5$ the refractive index of $\omega_2$.
\begin{figure}[H] 
\centerline{\begin{tabular}{cc}
\includegraphics[width=0.45\textwidth]{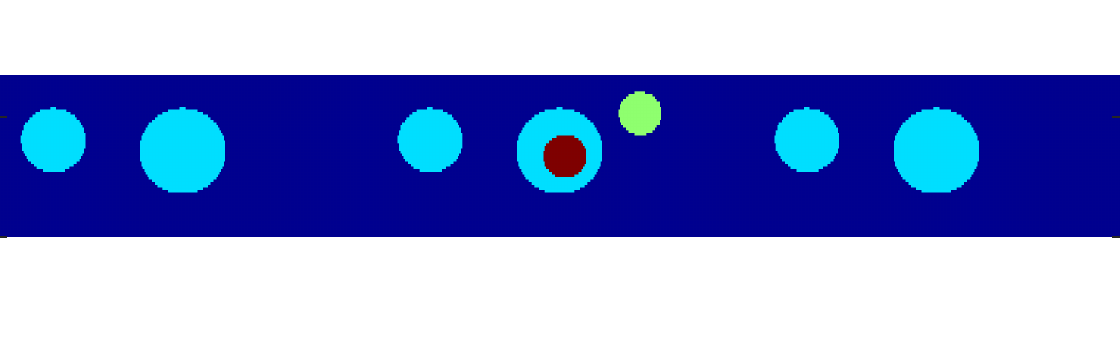} &
\includegraphics[width=0.45\textwidth]{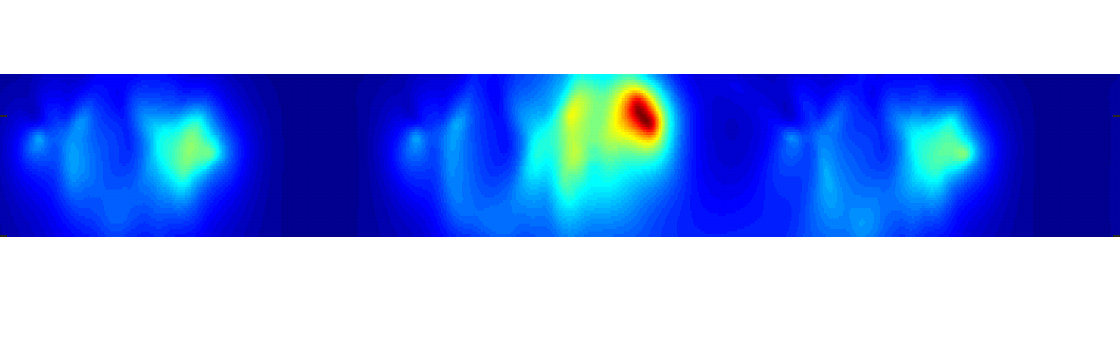}
\end{tabular}}
\caption{Left: The exact geometry for Example 5. Right: The reconstruction using $z \mapsto
 \I^\delta(z)$}
\label{IncluMixed_structure}
\end{figure}
\noindent
All our numerical examples validate the theoretical prediction provided by Theorem \ref{DiffImagFunc}. As already mention the case when the defect is entirely included in a component of the periodic background is ambiguous in the sense that the actual defective period  can not be determined. 

\noindent
We conclude by remarking that everything here can be adapted to the case when the background periodic layer is composed of inhomogeneities embedded in homogeneous media with constant refractive index $n_0\neq 1$. The only difference is in the choice in the fundamental solution which in this case should correspond to a piecewise homogeneous media instead of $n=1$. 

\section*{Acknowledgements}
The research of F. Cakoni is supported in part by AFOSR Grant FA9550-17-1-0147 and NSF Grant DMS- 1602802. The research of T-P. Nguyen is supported in part by NSF Grant DMS- 1602802.


\section*{References}
 \begin{thebibliography}{10}

\bibitem{Arens2010}
Tilo Arens.
\newblock Scattering by biperiodic layered media: The integral equation
  approach.
\newblock Habilitation Thesis, Universit{\"a}t Karlsruhe, 2010.

\bibitem{Arens2005}
Tilo Arens and Natalia Grinberg.
\newblock A complete factorization method for scattering by periodic
  structures.
\newblock {\em Computing}, 75:111--132, 2005.

\bibitem{Audib2015a}
Lorenzo Audibert.
\newblock {\em {Qualitative methods for heterogeneous media}}.
\newblock Theses, {Ecole Doctorale Polytechnique}, September 2015.

\bibitem{Audib2015}
Lorenzo Audibert, Alexandre Girard, and Houssem Haddar.
\newblock Identifying defects in an unknown background using differential
  measurements.
\newblock {\em Inverse Problems and Imaging}, 9(3):625--643, July 2015.

\bibitem{Audib2014}
Lorenzo Audibert and Houssem Haddar.
\newblock A generalized formulation of the linear sampling method with exact
  characterization of targets in terms of farfield measurements.
\newblock {\em Inverse Problems}, 30(3):035011, March 2014.

\bibitem{Bourg2014}
L~Bourgeois and S~Fliss.
\newblock On the identification of defects in a periodic waveguide from far
  field data.
\newblock {\em Inverse Problems}, 30(9):095004, September 2014.

\bibitem{CCH}
Fioralba Cakoni, David Colton, and Houssem Haddar.
\newblock {\em Inverse scattering theory and transmission eigenvalues},
  volume~88 of {\em CBMS-NSF Regional Conference Series in Applied
  Mathematics}.
\newblock Society for Industrial and Applied Mathematics (SIAM), Philadelphia,
  PA, 2016.

\bibitem{CK3}
David Colton and Rainer Kress.
\newblock {\em Inverse acoustic and electromagnetic scattering theory},
  volume~93 of {\em Applied Mathematical Sciences}.
\newblock Springer, New York, third edition, 2013.

\bibitem{Elsch2011a}
Johannes Elschner and Guanghui Hu.
\newblock Inverse scattering of elastic waves by periodic structures:
  uniqueness under the third or fourth kind boundary conditions.
\newblock {\em Methods and Applications of Analysis}, 18(2):215--244, 2011.

\bibitem{Thi-Phong2}
Houssem Haddar and Thi-Phong Nguyen.
\newblock {A volume integral method for solving scattering problems from
  locally perturbed infinite periodic layers}.
\newblock {\em {Applicable Analysis}}, pages 130 --158, 2016.

\bibitem{Thi-Phong3}
Houssem Haddar and Thi-Phong Nguyen.
\newblock {Sampling methods for reconstructing the geometry of a local
  perturbation in unknown periodic layers}.
\newblock {\em {Computers and Mathematics with Applications}},
  74(11):2831--2855, December 2017.

\bibitem{Kirsc2008}
A.~Kirsch and N.I. Grinberg.
\newblock {\em The Factorization Method for Inverse Problems}.
\newblock Oxford Lecture Series in Mathematics and its Applications 36. Oxford
  University Press, 2008.

\bibitem{kirschitp}
Andreas Kirsch.
\newblock A note on {S}ylvester's proof of discreteness of interior
  transmission eigenvalues.
\newblock {\em C. R. Math. Acad. Sci. Paris}, 354(4):377--382, 2016.

\bibitem{KL}
Andreas Kirsch and Armin Lechleiter.
\newblock The {L}imiting {A}bsorption {P}rinciple and a {R}adiation {C}ondition
  for the {S}cattering by a {P}eriodic {L}ayer.
\newblock {\em SIAM J. Math. Anal.}, 50(3):2536--2565, 2018.

\bibitem{Lechl2013b}
Armin Lechleiter and Dinh-Liem Nguyen.
\newblock Factorization {Method} for {Electromagnetic} {Inverse} {Scattering}
  from {Biperiodic} {Structures}.
\newblock {\em SIAM Journal on Imaging Sciences}, 6(2):1111--1139, June 2013.

\bibitem{armin}
Armin Lechleiter and Ruming Zhang.
\newblock Reconstruction of local perturbations in periodic surfaces.
\newblock {\em Inverse Problems}, 34(3):035006, 17, 2018.

\bibitem{nguye2012}
Dinh~Liem Nguyen.
\newblock {\em Spectral Methods for Direct and Inverse Scattering from Periodic
  Structures}.
\newblock PhD thesis, {Ecole Polytechnique X}, 2012.

\bibitem{tpnguyen}
Thi~Phong Nguyen.
\newblock {\em {Direct and inverse solvers for scattering problems from locally
  perturbed infinite periodic layers}}.
\newblock Theses, {Universit{\'e} Paris-Saclay}, January 2017.

\bibitem{Sylve2012}
John Sylvester.
\newblock Discreteness of {Transmission} {Eigenvalues} via {Upper} {Triangular}
  {Compact} {Operators}.
\newblock {\em SIAM Journal on Mathematical Analysis}, 44(1):341--354, January
  2012.

\end{thebibliography}
\end{document}